\documentclass{article}
\usepackage[T1]{fontenc}
\usepackage[latin9]{inputenc}
\usepackage{amssymb,amsthm,amssymb}
\usepackage{graphicx}
\usepackage{psfrag}
\usepackage{subfig}
\usepackage{psfrag}
\usepackage{graphics}

\newtheorem{thm}{Theorem}[section]

\newtheorem{lem}{Lemma}[section]
\newtheorem{defn}{Definition}[section]



\begin{document}
 \title{Two player game variant of the Erd\H{o}s-Szekeres problem}
\date{}

\newcommand{\atgen}{\symbol{'100}}
\author{
Parikshit Kolipaka \thanks{
Department of Computer Science and Automation,
Indian Institute of Science, Bangalore, India.
E-mail: \texttt{parikshit\atgen{}csa.iisc.ernet.in}
} 
\and
Sathish Govindarajan \thanks{
Department of Computer Science and Automation,
Indian Institute of Science, Bangalore, India.
E-mail: \texttt{gsat\atgen{}csa.iisc.ernet.in}
}
}
\maketitle
\thispagestyle{empty}
\begin{abstract}
The classical Erd\H{o}s-Szekeres theorem states that a convex $k$-gon exists in every sufficiently large point set.
 This problem has been well studied and finding tight asymptotic bounds is considered a challenging open problem.
Several  variants of the Erd\H{o}s-Szekeres problem have been posed and studied in the last two decades. The well studied variants include the empty convex $k$-gon  problem, convex $k$-gon with specified number of interior points 
 and the chromatic variant.

In this paper, we introduce the following two player game variant of the Erd\H{o}s-Szekeres problem: Consider a two player game
 where each player playing in alternate turns, place points in the plane.
 The objective of the game is to avoid the formation of the convex k-gon among the placed points.
The game ends when a convex k-gon is formed and the player who placed the last point loses the game.
In our paper we show a winning strategy for the player who plays second in the convex 5-gon game and the empty convex 5-gon game 
by considering convex layer configurations at each step. We prove that the game always 
ends in the 9th step by showing that the game reaches a specific set of configurations.
\end{abstract}

\section{Introduction}
\noindent The Erd\H{o}s-Szekeres problem is defined as follows: 
\textit{For any integer k, k \ensuremath{\ge}
3, determine the smallest positive integer N(k) such that any planar point set in general position that has at least N(k) points contains k points that are the vertices of a
convex k-gon.}

In 1935 Erd\H{o}s and Szekeres proved the finiteness of $N(k)$ using Ramsey theory \cite{Erdos35acombinatorial}.
There has been a series of improvements to bound the value of $N(k)$ and the current best known bounds 
are 2$^{k-2}$+ 1 \ensuremath{\le} $N(k)$ \ensuremath{\le} $\left(_{k-2}^{2k-5}\right)$+1 \cite{Erdos78somemore,Kleitman98findingconvex,Toth04theerdos-szekeres}.
 Erd\H{o}s and Szekeres conjectured that the current lower bound is tight. This conjecture has been proved for $k$ \ensuremath{\le} $6$.
$N(4)=5$ was shown by Klein and $N(5)= 9$ was shown by Kalbfeisch $et.al.$ \cite{kalbfeisch}.
 $N(6)=17$ has been proved by Szekeres and Peters using an combinatorial model of planar configurations  \cite{Szekeres_computersolution}.
See survey \cite{Morris00theerdos-szekeres,newsurvey} for a detailed description about the history of the problem  and its variants.

Erd\H{o}s asked the empty convex $k$-gon problem which is defined as follows:
\textit{ For any integer k, k \ensuremath{\ge}
3, determine the smallest positive integer $H(k)$ such that any planar point set in general position that has at least $H(k)$ points contains k points
 that are the vertices of a
an empty convex $k$-gon, i.e., the vertices of a convex $k$-gon containing
no points in its interior.}

It is easy to see that $H(4) = 5$. $H(5) = 10$ was proved by Harboth \cite{Harborth79konvexefunfecke}. Gerken \cite{Gerken05onempty} and Nicholes \cite{1340381}  proved independently the existence of a empty hexagon.
  The current best bounds on $H(6)$ are $30 \leq H(6) \leq 463$ \cite{kos,overmars}. Horton showed that $H(k)$ does not exist for any $k$, $k$ \ensuremath{\ge}
$7$ \cite{Horton83setswith}.

The Erd\H{o}s-Szekeres problem has been studied in higher dimensions \cite{Erdos35acombinatorial,Erdos78somemore,JOUR}.
Other related variants that are well studied are the convex $k$-gon with specified number of 
interior points \cite{Avisonthe1,on,Fevens04improvedlower,Hosono2005201,1667532}
 and the chromatic variant  \cite{Devillers03chromaticvariants,nonconvex}.
                
 We introduce the following two player game variant of Erd\H{o}s-Szekeres problem:

  \textit{Consider a two player game where each player playing in alternate turns, place points in the plane.
 The objective of the game is to avoid the formation of the convex k-gon among the placed points.
The game will end when a convex k-gon is formed and the player who placed the last point loses the game.}

In the game we assume that each player has infinite computational resources and hence place their points in optimal manner.

 Since $N(k)$ is finite, we know that the game will end at $N(k)$ number of steps. Can the game end before $N(k)$ steps?
 Define $N_{G}(k)$ as the minimum number of steps before the game ends. In this paper we focus on finding the exact value of $N_{G}(k)$. 
We denote the player who plays first as player 1 and the player who plays second as player 2.

We also consider the two player game for the empty convex $k$-gon and correspondingly define $H_{G}(k)$.
 It is easy to see that  $N_{G}(3)=H_{G}(3)=3$, $N_{G}(4)= H_{G}(4)=5$.

\subsection*{Results}
In this paper, we focus on the Erd\H{o}s-Szekeres two player game for $k = 5$. We show a winning strategy  for player 2 and prove that $N_{G}(5)=9$ and $H_{G}(5)=9$.
i.e., the game will end in the 9th step.

We consider convex layer configurations at each step and give a strategy for player 2 such that the game will
reach a specific set of configurations until the 8th step and finally we argue in the 9th step that a convex 5-gon or an empty convex 5-gon is formed.

\subsection*{Organization of the paper}
Section 2 contains the preliminaries and definitions that we will be using in the rest of the paper. Section 3  describes the proofs that
 the empty convex 5-gon game and convex 5-gon game ends in 9th step and player 2 wins
the game.

\section{Preliminaries and Definitions}

We assume in the rest of this paper that our point set $P$ is in general position, i.e., no 3 points of the point set are collinear.
We denote the convex hull of a point set $P$ as $conv(P)$ and the vertices of $conv(P)$ as $CH(P)$.

\begin{defn}\label{convex k-gon}\textbf{Convex k-gon:}
  is a convex polygon with k vertices.
\end{defn}

\begin{defn}\label{empty convex k-gon}\textbf{Empty Convex k-gon:}
  is a convex k-gon with
no points in its interior.
\end{defn}

\begin{defn}\label{Points in convex position}\textbf{Points in convex position:}
A point set $P$ is said to be in convex position if $CH(P)= P$. 
\end{defn}

\begin{defn}\label{convex layers}\textbf{Type of a point set P:}
A point set $P$ in of $type (i_{1},i_{2},...i_{k})$, $|P|=\sum|i_{k}|$, if $P_{1}= CH(P)$ is of size $i_{1}$, $P_{2}= CH(P \setminus P_{1})$ is of 
size $i_{2}$ \dots
\end{defn}

The $type$ of point set $P$ describes the sizes of the different 
convex layers of $P$. We denote $P_{1}$ as the first convex layer and $P_{2}$ as the second convex layer.

\begin{defn}\label{using}\textbf{U(i,j) of point set P:}
 Point set having $i$ points of the first convex layer of $P$  
and $j$ points of the second convex layer of  $P$.
\end{defn}

\begin{defn}\label{type1}\textbf{Type 1 Beam:}
  $A:BC$ denotes the region of the plane  formed by deleting triangle $ABC$ from the convex region in the plane 
bounded by the rays $\overrightarrow{AB}$ and $\overrightarrow{AC}$ (see figure \ref{typa}).  
\end{defn}

\begin{defn}\label{type2}\textbf{Type 2 Beam:}  $AB:CD$ denotes the region of the plane formed by deleting convex 4-gon $ABCD$ 
from the convex region in the plane
 bounded by the segment $\overline{AB}$ and the rays $\overrightarrow{AD}$ and $\overrightarrow{BC}$ (see figure \ref{typb}). 
\end{defn}

\begin{figure}[ht]
\begin{minipage}[b]{0.5\linewidth}
  \centering
   {\includegraphics[width=0.5\textwidth]{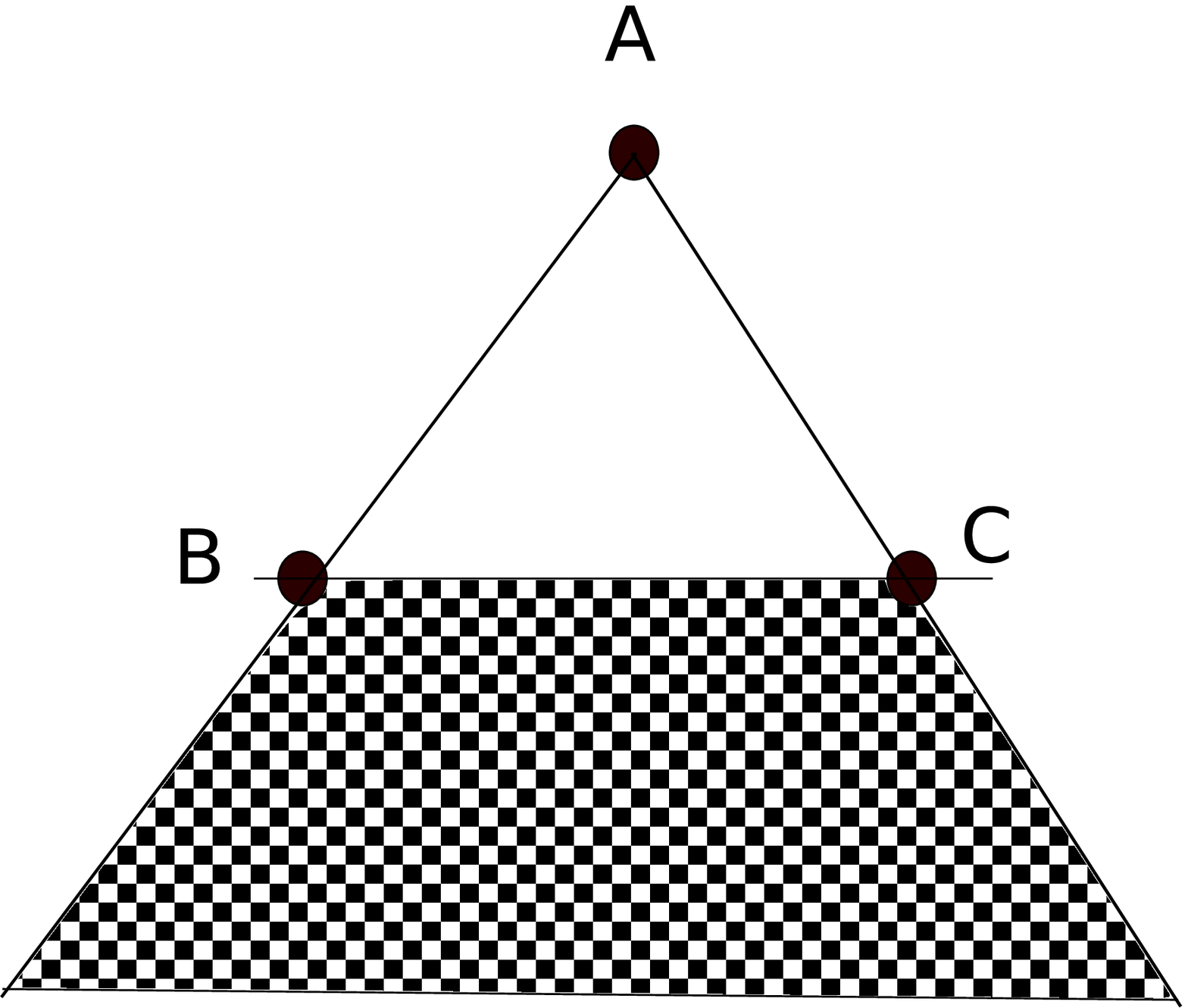}} 
   \hspace{0.3\textwidth}
   \caption{Type 1 beam}
    \label{typa}
\end{minipage}
\begin{minipage}[b]{0.5\linewidth}
  \centering
   {\includegraphics[width=0.5\textwidth]{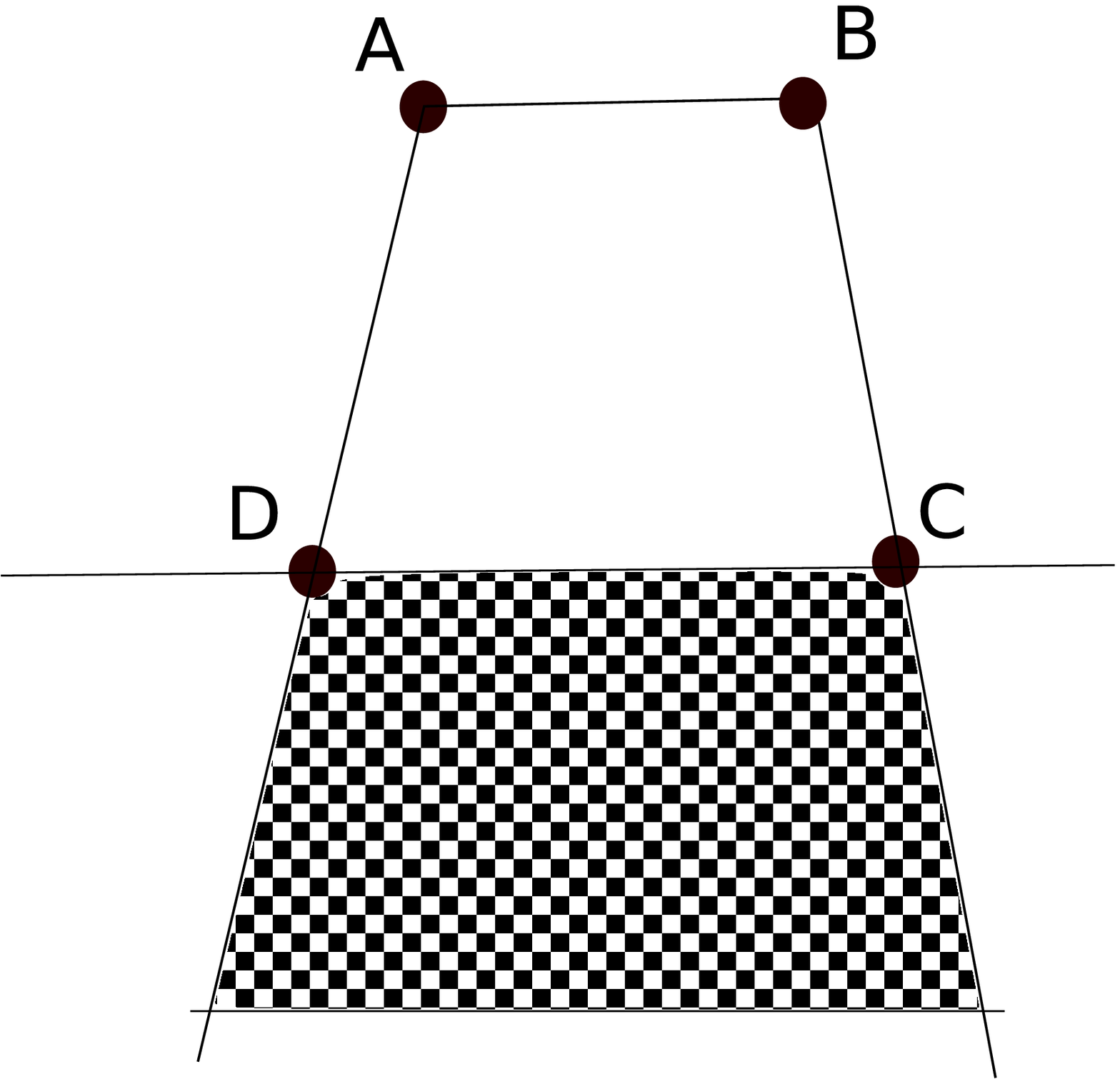}} 
   \hspace{0.3\textwidth}
   \caption{Type 2 beam}
    \label{typb}
\end{minipage}
\end{figure}
 \begin{figure}
  \centering
  {\label{4gonregionse}\includegraphics[width=0.5\textwidth]{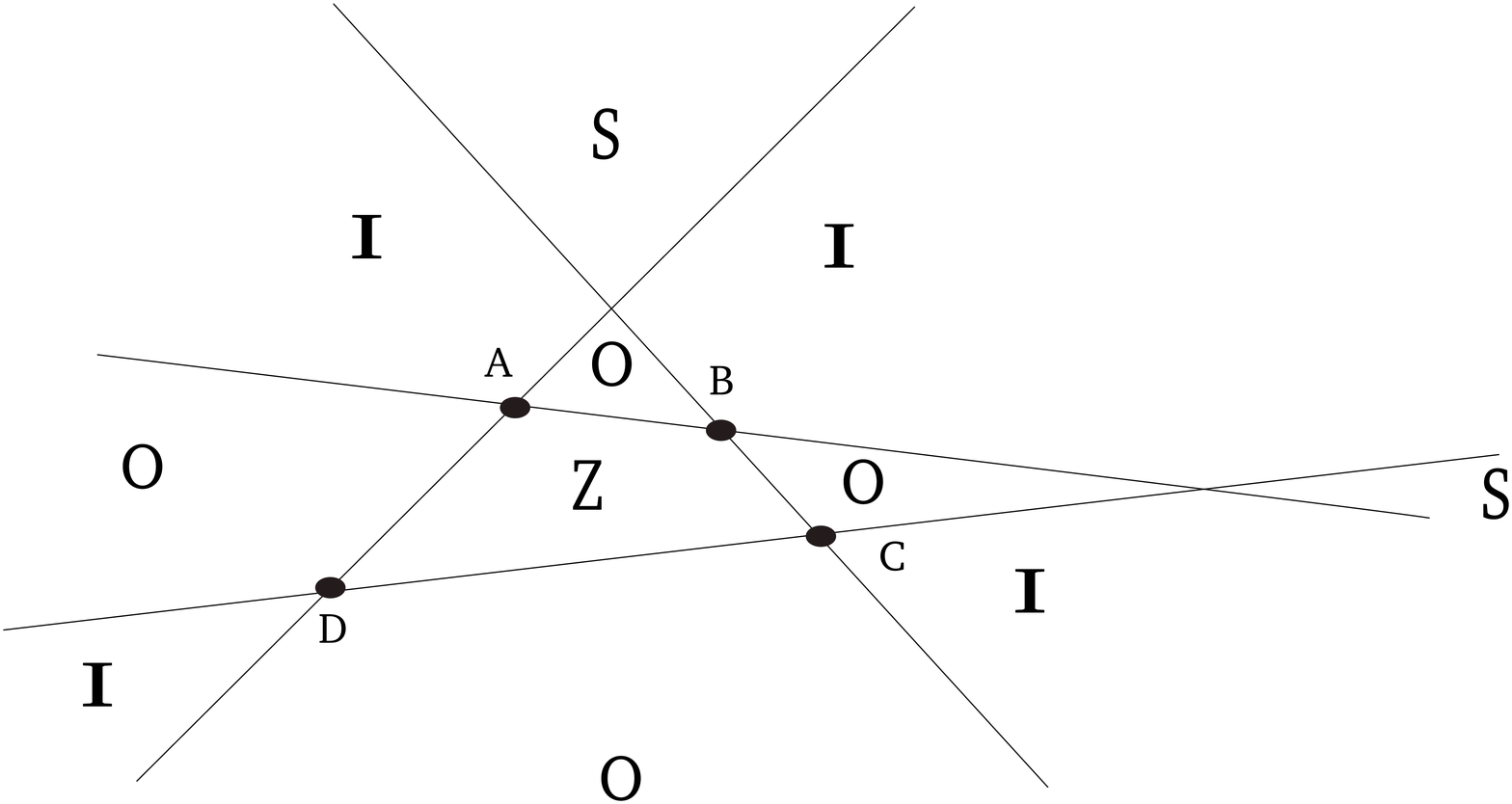}} 
  \hspace{0.1\textwidth}
  \caption{Divison of convex 4-gon into regions}
\end{figure}

 Let $A,B,C,D$ be 4 points in convex position.
\begin{defn}\label{coongrm}\textbf{Regions of empty convex 4-gon:}
The $ABCD$ convex 4-gon  divides the plane into 4 types of regions $I,O,S,Z$ (see $figure$ 3). 
$O$ region is the region such that any point in it along with $ABCD$  forms a convex 5-gon. 
$I$ and $Z$ regions are the regions such that any point in these regions along with $ABCD$ forms a point set of $type(4,1)$. $I$ region is outside the convex 4-gon and
 $Z$ region is inside the convex 4-gon.
$S$ region is the region such that any point in it along with $ABCD$  forms a point set of $type(3,2)$.
\end{defn}

We define the following set of point configurations that will be used in our proofs (see $figure$ \ref{tree}). 
 
\begin{defn}\label{config4}\textbf{Configuration 4:}
A 4 point set forming a parallelogram.
\end{defn}
\begin{defn}\label{config5.1}\textbf{Configuration 5.1:}
A 5 point set of $type(4,1)$ where the 4 points in the first convex layer form a parallelogram.
\end{defn}
\begin{defn}\label{config5.2}\textbf{Configuration 5.2:}
A 5 point set of $type(4,1)$ where the 3 points in the first convex layer with the 1 point in the interior form a parallelogram.
\end{defn}

 \begin{figure}
  \centering
  {\includegraphics[width=1.0\textwidth]{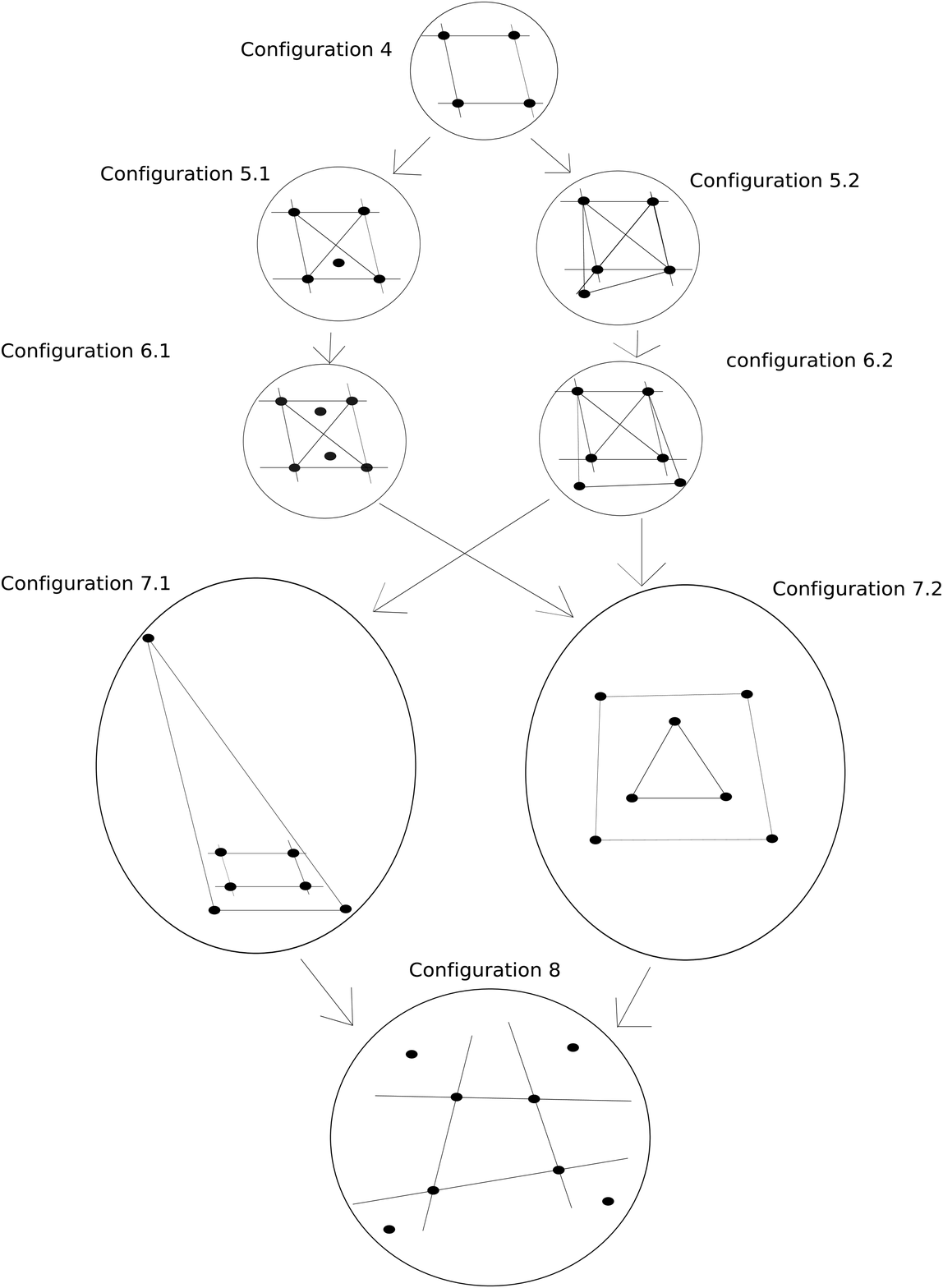}} 
  \hspace{1.0\textwidth}
  \caption{Game tree for the convex 5-gon and empty convex 5-gon}
\label{tree}
 \end{figure}

\begin{defn}\label{config6.1}\textbf{Configuration 6.1:} A  6 point set  of $type(4,2)$ where 
4 points in the first convex layer form a parallelogram and the 2 points inside the parrallelogram are symmetrically placed in opposite triangles formed by diagonals of parallelogram. 
\end{defn}

\begin{defn}\label{config6.2}\textbf{Configuration 6.2:} A  6 point set of $type(4,2)$ where
4 points in the first convex layer form a trapezoid and the 2 points inside the trapezoid are symmetrically placed 
 in  opposite triangles formed by the diagonals.
\end{defn} 

\begin{defn}\label{config7.1}\textbf{Configuration 7.1:}
 A 7 point set of $type(3,4)$  where  the 3 points of the first convex layer are in  different $I$ regions 
of the  parralleogram formed by the 4 points in the second convex layer.
\end{defn}

\begin{defn}\label{config7.2}\textbf{Configuration 7.2:}
 A 7 point set of $type(4,3)$  such that it does not have an empty convex 5-gon.
\end{defn}
 \begin{defn}\label{config8}\textbf{Configuration 8:}
 An  8 point set of $type(4,4)$  where  the 4 points of the first convex layer are placed such that each point lies in  different $I$ region 
of the  convex 4-gon of the second convex layer.
\end{defn}

Note that the above configurations of $i$ points  $4 \leq i \leq 8 $ do not contain a empty convex 5-gon and all configurations, except configuration 7.2 
and 8 do not contain convex 5-gon.  
\section{Game for the empty convex 5-gon and convex 5-gon}
 In this section, we show that the two player game for the empty convex 5-gon  and the convex 5-gon  ends in  9 moves and the second player has a winning strategy. 
\subsection*{Overview of our proof and player 2's strategy to win the empty convex 5-gon game and the convex 5-gon game.}

In our game, Player 1 plays in the odd steps (1st, 3rd, \dots) and Player 2 plays in the even steps (2nd, 4th, \dots).
 In player 1's turn, we argue that any point added without forming an convex 5-gon or empty convex 5-gon will always result 
in specific configurations. In player 2's turn we show a feasible region where if the point is placed will result in specific configurations that are
favorable for player 2.

We now describe the winning strategy for player 2:
Player 2  will  place the point in the 4th step such that the resultant point set forms a parallelogram (configuration 4). 
In the  6th step, we show that there exists a feasible region in both configuration 5.1 and 5.2 where the 6th point is placed, so that it will reach 
configuration 6.1 or configuration 6.2. Similarly in the 8th step we show that there will exist a feasible region in configuration 7.1 and 7.2  where the 8th point 
is placed such that the resultant point set is configuration 8.

In player 1's turn we show that any point added by player 1 without forming an convex k-gon or empty convex k-gon results in configuration 5.1 or 5.2 (5th step) and in configuration 7.1 or 7.2 (7th step). 

Finally, we argue that any point added to the configuration 8 will result in the formation of a convex 5-gon or empty convex 5-gon and player 2 will always win
 in the 9th step.

Thus, any convex 5-gon/empty convex 5-gon game follows a path in the game tree shown in $figure$ \ref{tree}.
The proofs for the convex 5-gon game and the empty convex 5-gon are similar, so we have combined them.
 
In the odd step (player 1's turn to place the point) the feasible regions that are formed in the convex 5-gon game
 are a subset of the feasible regions that are formed in the empty convex 5-gon game. This is because the regions which 
are not covered by the O regions of the convex 4-gons are a subset of the regions which are not covered by the O regions of the empty convex 4-gons.
 So in the odd step we give the proof for the empty convex 5-gon game and it is sufficient for both the games.

In the even step (player 2's turn to place the point) we give the proof for the convex 5-gon game by showing a feasible region where player 2 
places the point and this is sufficient for both the games because a feasible region in the convex 5-gon game is also a feasible region in the empty convex
5-gon game.

In the empty convex 5-gon game, we give a separate proof for point configurations which have a convex 5-gon that is not empty.

\subsection{Proof for  $H_{G}(5)=9$ and $N_{G}(5)=9$}
We make the following observations on the number of points that can contained in type 1 and type 2 beams without forming an empty convex 5-gon. We will assume that the points in the beam are in convex position with 
the points that form the beam. Thus, if the point set does not contain a convex 5-gon, type 1 beam has atmost 1 point and type 2 beam does not contain any point.

\begin{figure}
  \centering
  {\includegraphics[width=0.3\textwidth]{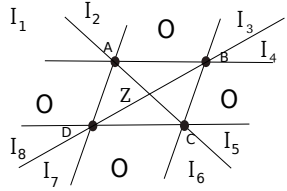}} 
  \hspace{0.1\textwidth}
  \caption{Divison of parallelogram regions}
  \label{parrallelogram}
\end{figure}

\begin{lem}
\label{lem1}
The game for the convex 5-gon and the empty convex 5-gon will always reach either configuration 6.1 or configuration 6.2 at the end of 6th step.
 
\end{lem}
\begin{proof} A triangle is formed by the first 3 points of the game. The 4th point is placed in such a manner that the resultant point set forms a parallelogram $ABCD$ (configuration 4).
For the 5th step, we divide the regions of the parallelogram into $I, O, Z$ regions as shown in $figure$  \ref{parrallelogram}.
If a point is placed in the $O$ region it forms an empty convex 5-gon with the four points of the parallelogram.
 So the only feasible regions where a point is placed  are the $I$ regions and $Z$ region.

\begin{figure}[ht]
\begin{minipage}[b]{0.5\linewidth}
  \centering
   {\includegraphics[width=0.65\textwidth]{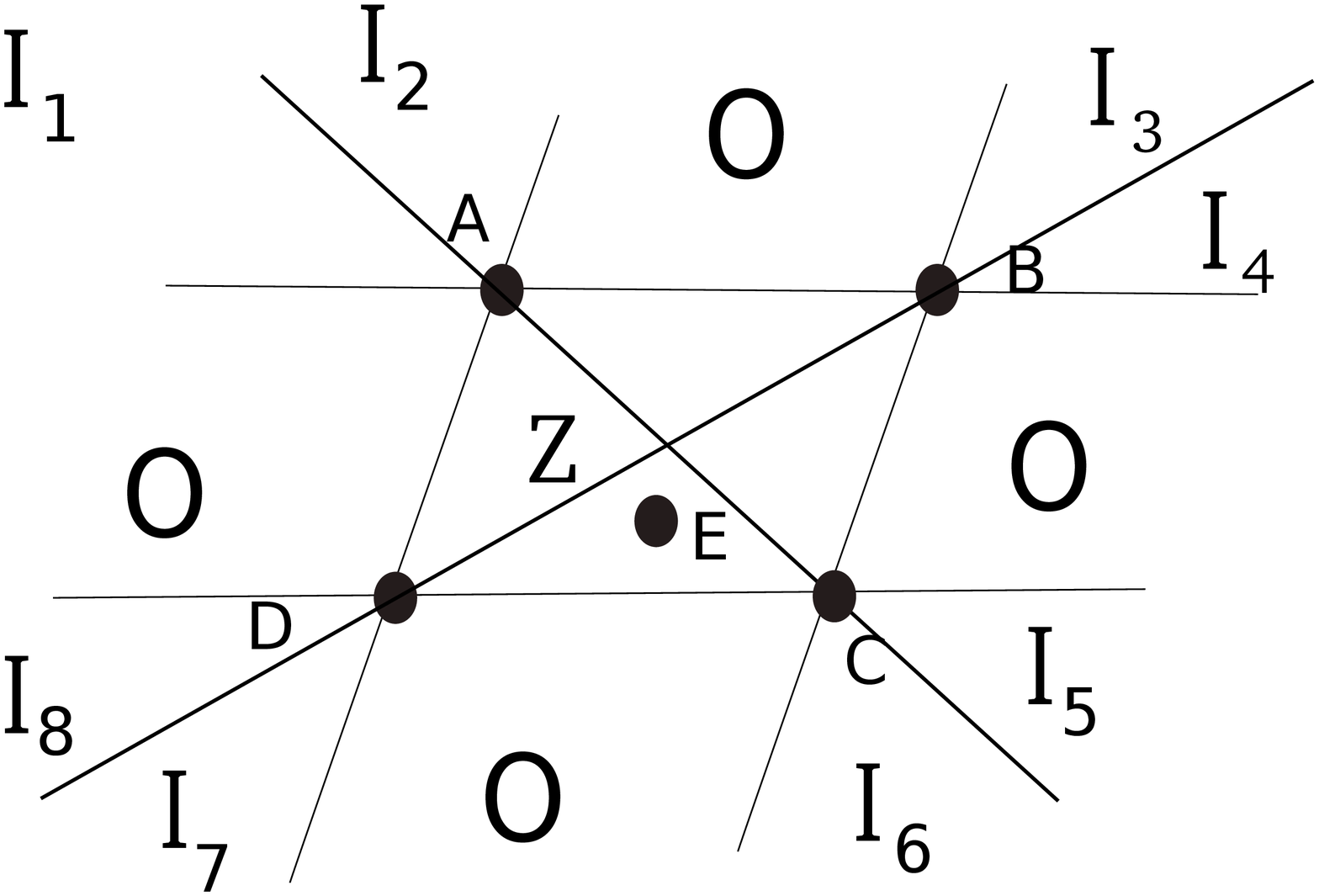}} 
   \hspace{0.18\textwidth}
   \caption{Configuration 5.1}
\label{config5.1}
\end{minipage}
\begin{minipage}[b]{0.5\linewidth}
  \centering
   {\includegraphics[width=0.67\textwidth]{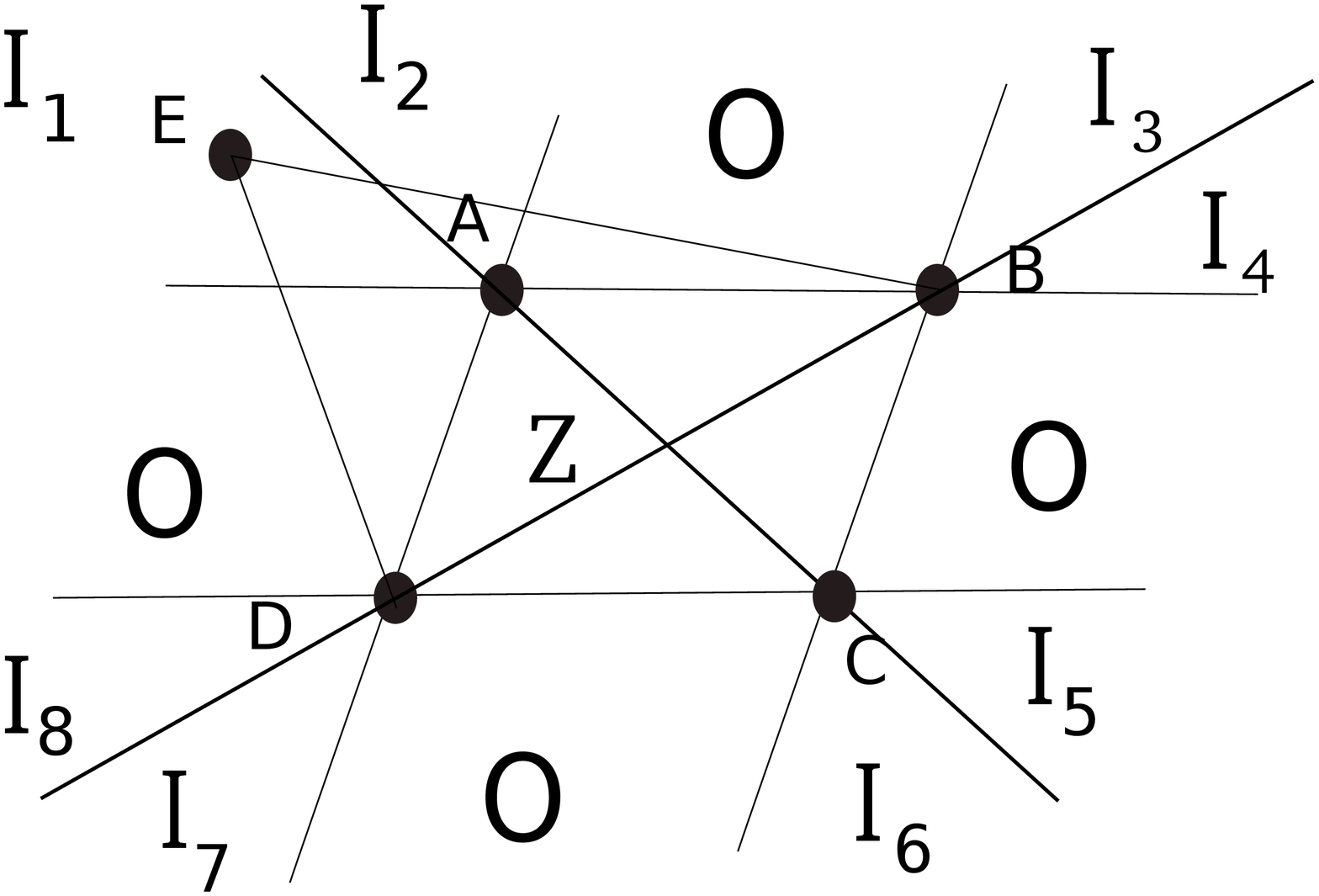}} 
   \hspace{0.3\textwidth}
   \caption{Configuration 5.2}
\label{config5.2}
\end{minipage}
\end{figure}

If the 5th point, say $E$, is placed in the interior of the parallelogram i.e., $Z$ region, the resultant point set forms configuration 5.1 (see $figure$ \ref{config5.1}).
If $E$ is placed in $I$ region the resultant point set forms configuration 5.2 (see $figure$ \ref{config5.2}).

\begin{figure}[ht]
\begin{minipage}[b]{0.5\linewidth}
  \centering
   {\includegraphics[width=0.67\textwidth]{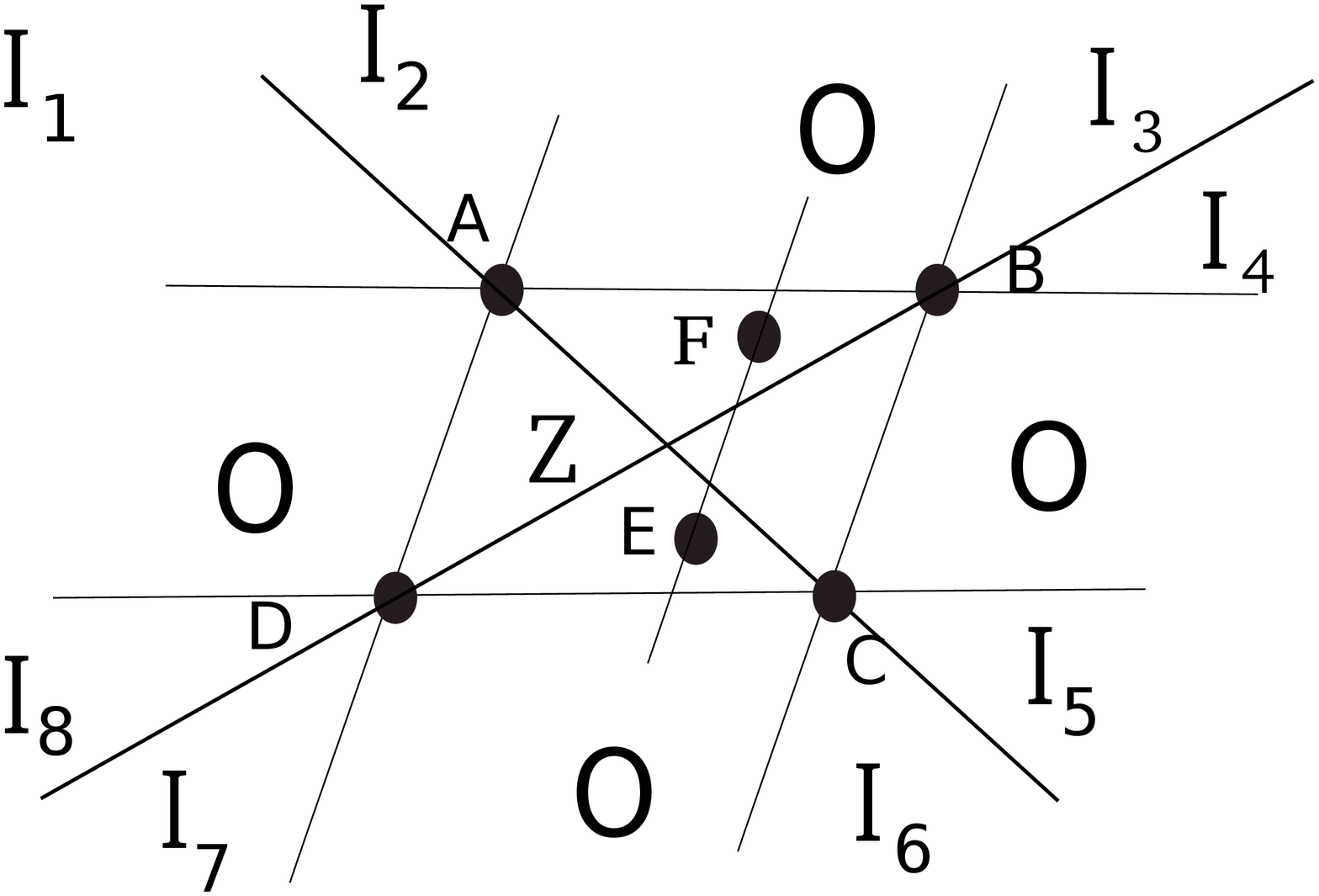}} 
   \hspace{0.3\textwidth}
   \caption{Configuration 6.1}
\label{config6.1}
\end{minipage}
\begin{minipage}[b]{0.5\linewidth}
  \centering
   {\includegraphics[width=0.68\textwidth]{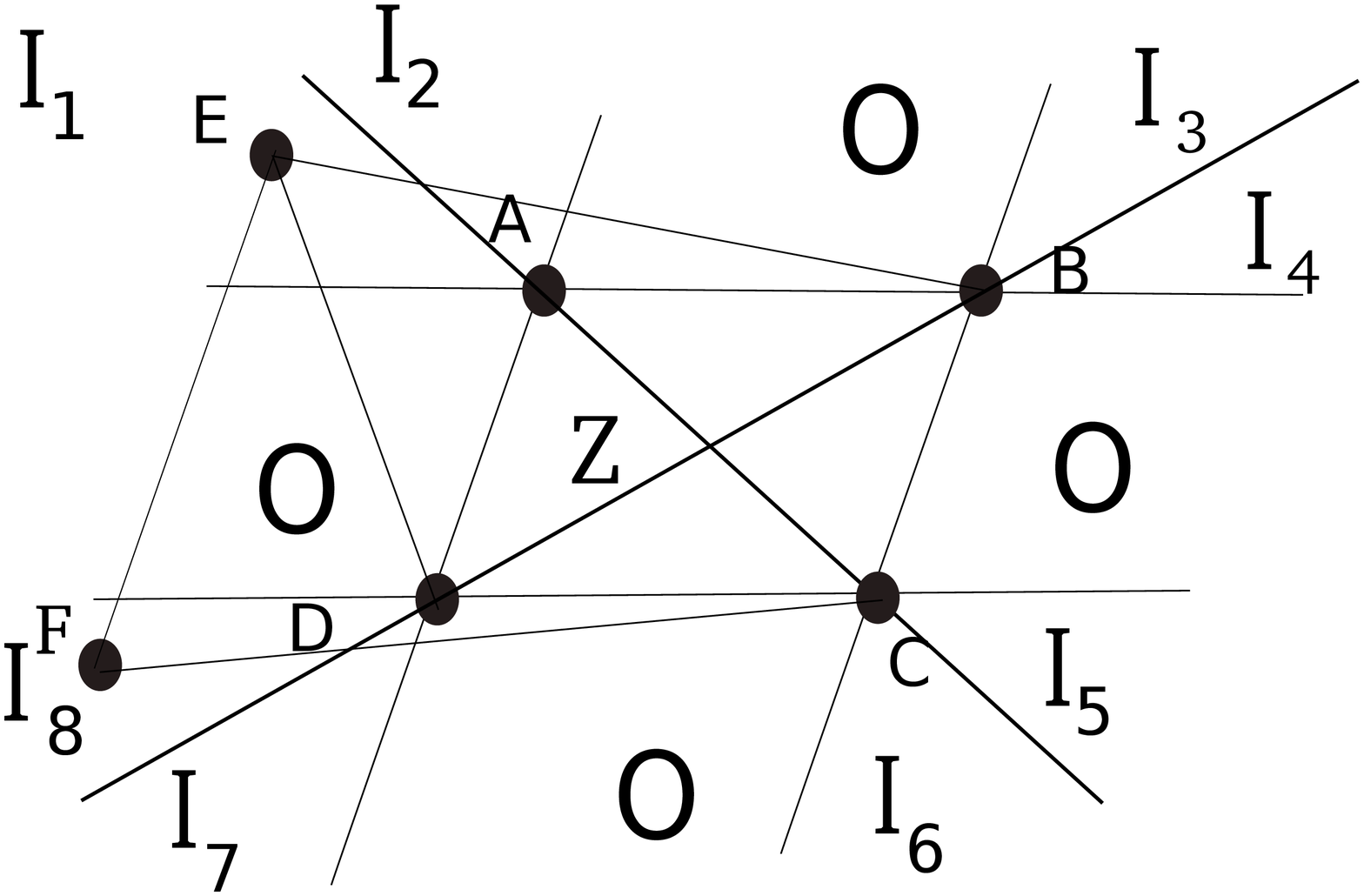}} 
   \hspace{0.3\textwidth}
   \caption{Configuration 6.2}
\label{config6.2}
\end{minipage}
\end{figure}

In the 6th step we show a feasible region in configuration 5.1 and configuration 5.2  where a point is added 
such that the resultant point set formed is either configuration 6.1 or configuration 6.2.

If the 5 point set formed is  configuration 5.1, then the 6th point, say $F$ is placed 
in a triangle (formed by the diagonals of the parallelogram)
 that is opposite to the triangle that has $E$ such that $EF$ is parallel to $AD$ and this forms configuration 6.1
(see $figure$ \ref{config6.1}).

Now we argue for configuration 5.2. Let us suppose that $E$ has been placed in $I_{1}$ region (see $figure$ \ref{config5.2}). 
The 6th point, say $F$, is placed in $I_{8}$ region  such 
that $EF$ is parallel to $AD$ (see $figure$  \ref{config6.2}). Let us call this as $\left\{ I_{1},I_{8}\right\}$  
point placement. By symmetry $ \left\{ I_{2},I_{3}\right\} ,\left\{ I_{4},I_{5}\right\} ,\left\{ I_{6},I_{7}\right\}$
 point placements are similar and they form configuration 6.2.
\end{proof}

\begin{lem}
\label{lem2}
  Any point added to either configuration 6.1 or configuration 6.2 without forming an convex 5-gon or an empty convex 5-gon,  
results in either configuration 7.1 or configuration 7.2.
\end{lem}
\begin{proof}
We consider 2 cases corresponding to configuration 6.1 and configuration 6.2. We denote the regions of configuration 6.1 and 6.2  as $O$ region
if it is infeasible (point added in this region forms an empty convex 5-gon) and $I,S,Z$ region if there exists a feasible region in their interior.  
 
\begin{figure}
  \centering
   {\includegraphics[width=0.42\textwidth]{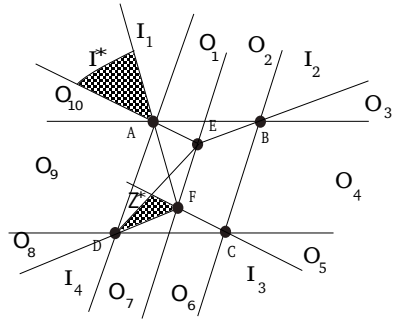}} 
   \hspace{0.1\textwidth}
   \caption{Configuration 6.1 divided into regions}
\label{p4}
\end{figure}

\textbf{Case 1 (configuration 6.1):} Let us consider all the empty convex 4-gons of configuration 6.1 (see $figure$ \ref{p4}). 
The regions that are not covered by the $O$ regions of the empty convex 4-gons are precisely the regions where a point can be
 added without forming a empty convex 5-gon.

The convex 4-gons that are formed by $U(4,0),U(3,1)$ of configuration 6.1 are not empty, so they are not considered.
 The empty convex 4-gons are of the form  $U(2,2)$ of configuration 6.1. 

\textit{U(2,2) of configuration 6.1:}
$EFDA, EFCB, BEDF, AECF$ are the empty convex 4-gons. Consider $EFDA$ empty convex 4-gon. The regions that are covered by the $O$ regions of $EFDA$ are $O_{1},O_{7},O_{8},O_{9},O_{10}$. 
 Similarly the regions that are covered by the $O$ regions of $EFCB$ empty convex 4-gon are $O_{2},O_{3},O_{4},O_{5},O_{6}$. 
  The $O$ regions of $BEDF, AECF$ empty convex 4-gons does not cover the entire region of $Z$ and $I_{k}$ where $k=1,2,3,4$.
Thus, the only feasible regions of configuration 6.1 are portions of $I$ and $Z$ region.
 From $figure$ \ref{p4}, we can verify that 
region $I^{*}$ of $I_{1}$ and $Z^{*}$ of $Z$ is not covered by any of the $O$ regions of
the empty convex 4-gons. Figure \ref{p4} shows $I^{*}$ only in one  $I$ region and $Z^{*}$ in $Z$ region.
Symmetrically, there are feasible regions $I^{*}$ in the other $I$ regions and $Z^{*}$ in $Z$ region. Any point added in $I^{*}$ or $Z^{*}$ region will result in configuration 7.2.
\newline 
\textbf{Case 2 (configuration 6.2):} Let us consider the empty convex 4-gons of configuration 6.2 (see $figure$ \ref{1o}). The convex 4-gons that are formed by $U(4,0),U(3,1)$ of configuration 6.2 are not empty, so they are not considered.
 The empty convex 4-gons are of the form  $U(2,2)$ of configuration 6.2 (see $figure$ \ref{1o}).
 \begin{figure}
  \centering
   {\includegraphics[width=0.6\textwidth]{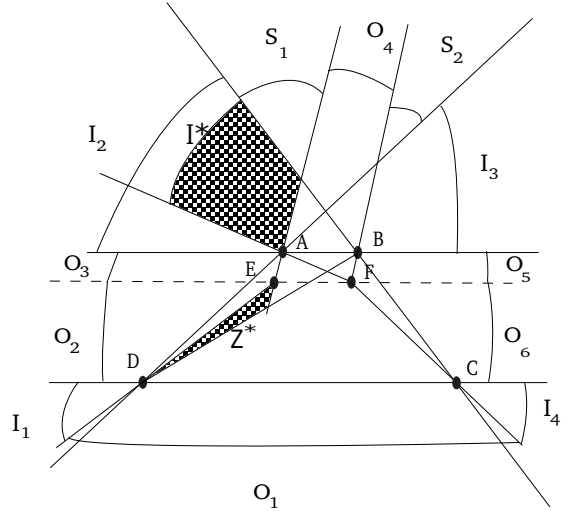}} 
   \hspace{0.1\textwidth}
   \caption{Configuration 6.2 divided into regions}
\label{1o}
\end{figure}

\textit{U(2,2) of configuration 6.2:}
$EFCD, ABFE, AFCE, BEDF$ are the empty convex 4-gons.
Consider $EFCD$ empty convex 4-gon. The regions that are completely covered by the $O$ regions of $EFCD$ are 
$O_{2},O_{6},O_{1}$. 
 Similarly the regions covered by the $O$ regions of $ABFE$ empty convex 4-gon are $O_{3},O_{5},O_{4}$.
The O-regions of  $AFCE$,$BEDF$ empty convex 4-gons does not cover the entire regions of  $Z$ and $I_{k}$ where $k=1,2,3,4$.
Thus, the only feasible regions of configuration 6.2 is $S_{1},S_{2}$ and portions of $I$ and $Z$ regions. 
From $figure$ \ref{1o}, we can verify that region $I^{*}$ of $I_{2}$  is not covered by any of the $O$ regions of these convex 4-gons.
Figure \ref{1o} shows $I^{*}$ only in one  $I$ region and $Z^{*}$ in $Z$ region.
Symmetrically, there are feasible regions $I^{*}$ in the other $I$ regions and $Z^{*}$ in $Z$ region.

If the point is placed in $I^{*}$ of $I_{k}$ where $k=1,2,3,4$ or $Z^{*}$ of $Z$,
 the resultant point set is configuration 7.2. 
If the point is placed in $S_{1}$ or $S_{2}$ the resultant point set is configuration 7.1.
\end{proof}

Now we prove that given a set of 7 points in configuration 7.1 or 7.2 there exists a feasible region
  where a point is added such that the resultant point set is configuration 8.

\begin{lem}
\label{le3}
There exists a feasible region in configuration 7.1 and 7.2 such that a point added in the feasible regions results in configuration 8.
\end{lem}
\begin{proof}
We consider 2 cases corresponding to configuration 7.1 and configuration 7.2.
 \begin{figure}
  \centering
   {\includegraphics[width=0.45\textwidth]{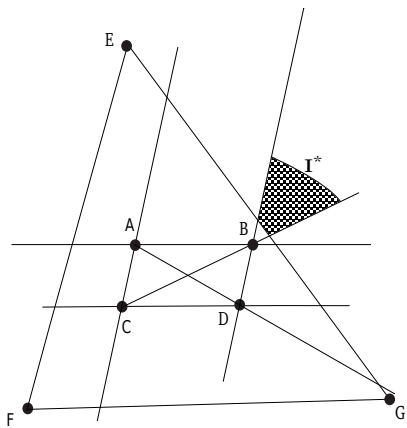}} 
   \hspace{0.1\textwidth}
   \caption{Configuration 7.1}
\label{lq1}
\end{figure}

\textbf{Case 1 (configuration 7.1):}
 Let us consider all the convex 4-gons of configuration 7.1. The regions that are
 not covered by the $O$ regions of convex 4-gons are precisely the regions where a point is added without forming a convex 5-gon. Consider the region $I^{*}$ as shown in $figure$ \ref{lq1}. We will show that this is a feasible region.
 The convex 4-gons are of the form $U(2,2),U(1,3),U(0,4)$ of configuration 7.1. 

\textit{U(0,4) of configuration 7.1:}
$ABCD$ is the only convex 4-gon of this type and has $I^{*}$ in its $I$ region.

\textit{U(1,3) of configuration 7.1:}
 $EBDA,FDBC,ADGC,ABGC,DCEB,ABDF$ are the convex 4-gons of this type and these have $I^{*}$ in their $I$ regions.

\textit{U(2,2) of configuration 7.1:}
$FCDG,EACF,ABGF,BDFE,ADGF,FCBG$ are the convex 4-gons of this type and these have $I^{*}$ in their $I$ regions.
 Hence $I^{*}$ is feasible region where a point is added and the resultant point set that is formed is configuration 8.
 
\textbf{Case 2 (configuration 7.2)}: Configuration 7.2 is a point set of (4,3) convex layer configuration without a convex 5-gon.   
First, we give the following characterization of configuration 7.2: 
 \begin{figure}
  \centering
   {\includegraphics[width=0.3\textwidth]{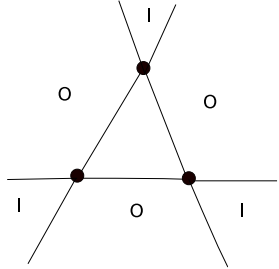}} 
   \hspace{0.1\textwidth}
   \caption{Divison of triangle into regions}
\label{12}
\end{figure}
\textit{An $I$ region or an $O$ region of the inner triangle cannot have more than 1 point (see $figure$ \ref{12})}.  

 Consider an $O$ region. If an $O$ region has 2 points then there is an empty convex 5-gon being formed by these 2 points with the 3 points of the triangle. 
 Hence an $O$ region cannot have more than 1 point.  
Consider an $I$ region. If an $I$ region has 2 points then no other $I$ region or 
the adjacent $O$ regions has any points because the 2 points in this $I$ region along with the third point in the $I$ region
 or in the adjacent $O$ region will form an empty convex 5-gon with the side of the triangle. Thus if an $I$ region has  2 points then the $O$ region that is 
opposite has to have the other 2 points which leads to the formation of an empty convex 5-gon. 

Based on the above constraints, the only 
possible (4,3) convex layer configuration is $(I,I,O,O)$ (2 points each in different $I$ regions, 2 points each in different $O$ regions).

Let us assume that point $G$ lies in region $X_{1}$. The case when the point $G$ is in $X_{2}$ can be argued in a similar fashion (see $figure$ \ref{lq5},\ref{lq7}).
 We have 2 cases corresponding to $EBG$ being a anti-clockwise turn or clockwise turn.

Type 1 (when $G$ lies in $X_{1}$ and $EBG$ is a anti-clockwise turn):
 We will show that $I^{*}$ (triangle region bounded  by the  $CG$,$FA$ and $EA$) is a feasible region (see $figure$ \ref{lq4}).
Consider the  convex 4-gons of this configuration.
The convex 4-gons are either of the type $U(2,2)$ ,$U(1,3)$ or $U(3,1)$.
\begin{figure}
  \centering
   {\includegraphics[width=0.53\textwidth]{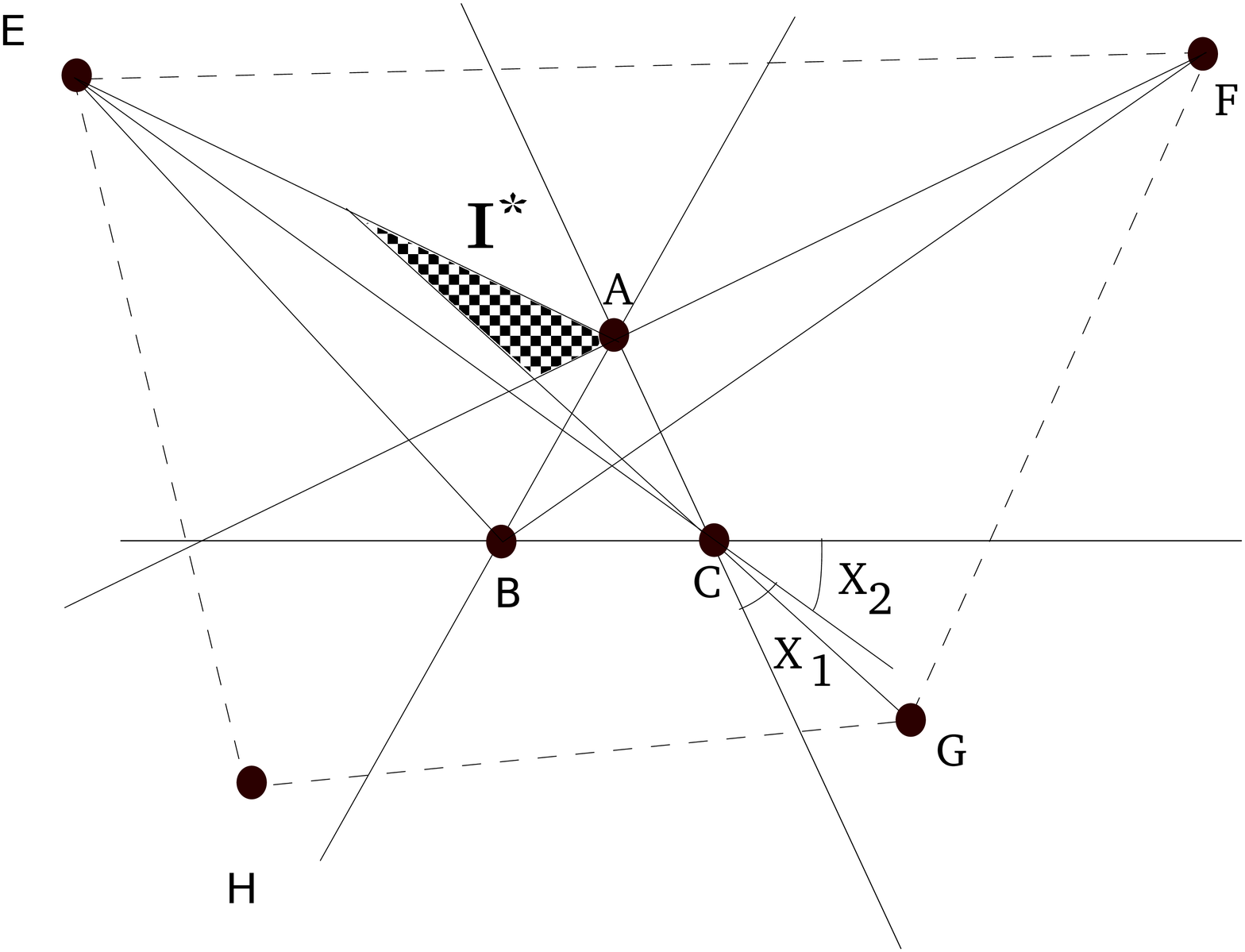}} 
   \hspace{0.1\textwidth}
    \caption{Configuration 7.2 (I,I,O,O) G is in X1 and EBG is anticlockwise turn}
\label{lq4}
\end{figure}

\textit{U(1,3) of Type1 Configuration 7.2 (I,I,O,O):}
$EACB,FABC$ are the 2 convex 4-gons of this type. $I^{*}$ region is in the interior of $EACB$ and it is in the $I$ region of $FABC$.
Thus $I^{*}$ is feasible for these convex 4-gons.  

\textit{U(2,2) of Type1 Configuration 7.2 (I,I,O,O):}
$EABH,FACG,ECGB,EACH$ $,FABG,BCFE,AHCF,EBGA$ and $HBFC$ or $AFBH$ based upon the position of $H$ in the $I$ region are the convex 4-gons of this type. It
is easy to see that 
 $I^{*}$ region is in the interior of $EABH,EACH,BCFE,EBGA$ and $I^{*}$ region is in the $I$ region of $FACG,ECGB,FABG,AHCF$. If $FAH$ is an anti clockwise turn and 
$FCH$ is clockwise turn,  then either $HBFC$ or $AFBH$ is the  convex 4-gon formed based upon the position of $H$ in the 
$I$ region of the triangle and both these  convex 4-gons have $I^{*}$ in their $I$ region. Thus $I^{*}$ is feasible for these convex 4-gons. 

\textit{U(3,1) of Type1 Configuration 7.2 (I,I,O,O):}
$EAGH,FAHG ,FBHG,ECGH,$ $EHCF,EBGF$ are the convex 4-gons of this type and these 4-gons have $I^{*}$ either in their $I$ region or in their $Z$ region.
If $FAH$ is an clockwise turn  then $EHAF$ is a convex 4-gon of this type having  $I^{*}$ region inside.
If $FAH$ is an anti clockwise turn and $FCH$ is anticlockwise turn then $FCHG$ is the  convex 4-gon is of this type and has $I^{*}$ in its $I$ region.
Thus $I^{*}$ is feasible for these convex 4-gons.                                                                                                     
Any point added to $I^{*}$ results in configuration 8.

Type 2 (when $G$ lies in $X_{1}$ and $EBG$ is a clockwise turn):
Consider the  convex 4-gons of this configuration.
The convex 4-gons are either of the type $U(2,2)$ ,$U(1,3)$ or $U(3,1)$ (see $figure$ \ref{lq6}).

\textit{U(3,1) of Type2 Configuration 7.2 (I,I,O,O):}  The convex 4-gons are the same as U(3,1) of Type1 Configuration 7.2 (I,I,O,O) 
except that we have $EBGH$ instead of $EBGF$. $EBGH$ has $I^{*}$  in its $I$ region.
 
\textit{U(1,3) of Type2 Configuration 7.2 (I,I,O,O):}
This is same as U(1,3) of Type1 Configuration 7.2 (I,I,O,O).

\textit{U(2,2) of Type2 Configuration 7.2 (I,I,O,O):}
This is a subset of U(2,2) of Type1 Configuration 7.2 (I,I,O,O).

For the empty convex 5-gon game, the following is a valid (4,3) configuration  containing a non empty convex 5-gon.
\begin{figure}
  \centering
   {\includegraphics[width=0.53\textwidth]{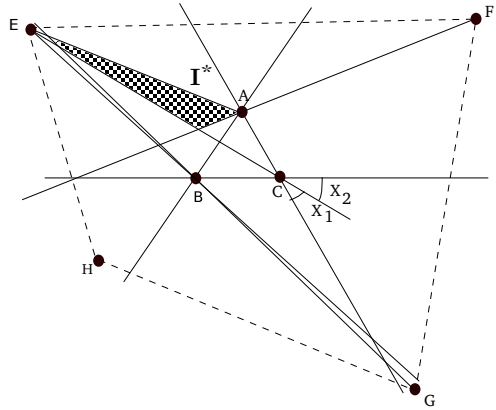}} 
   \hspace{0.1\textwidth}
   \caption{Configuration 7.2 (I,I,O,O) G is in X1 and EBG is clockwise turn}
\label{lq6}
\end{figure}

\begin{figure}
  \centering
   {\includegraphics[width=0.53\textwidth]{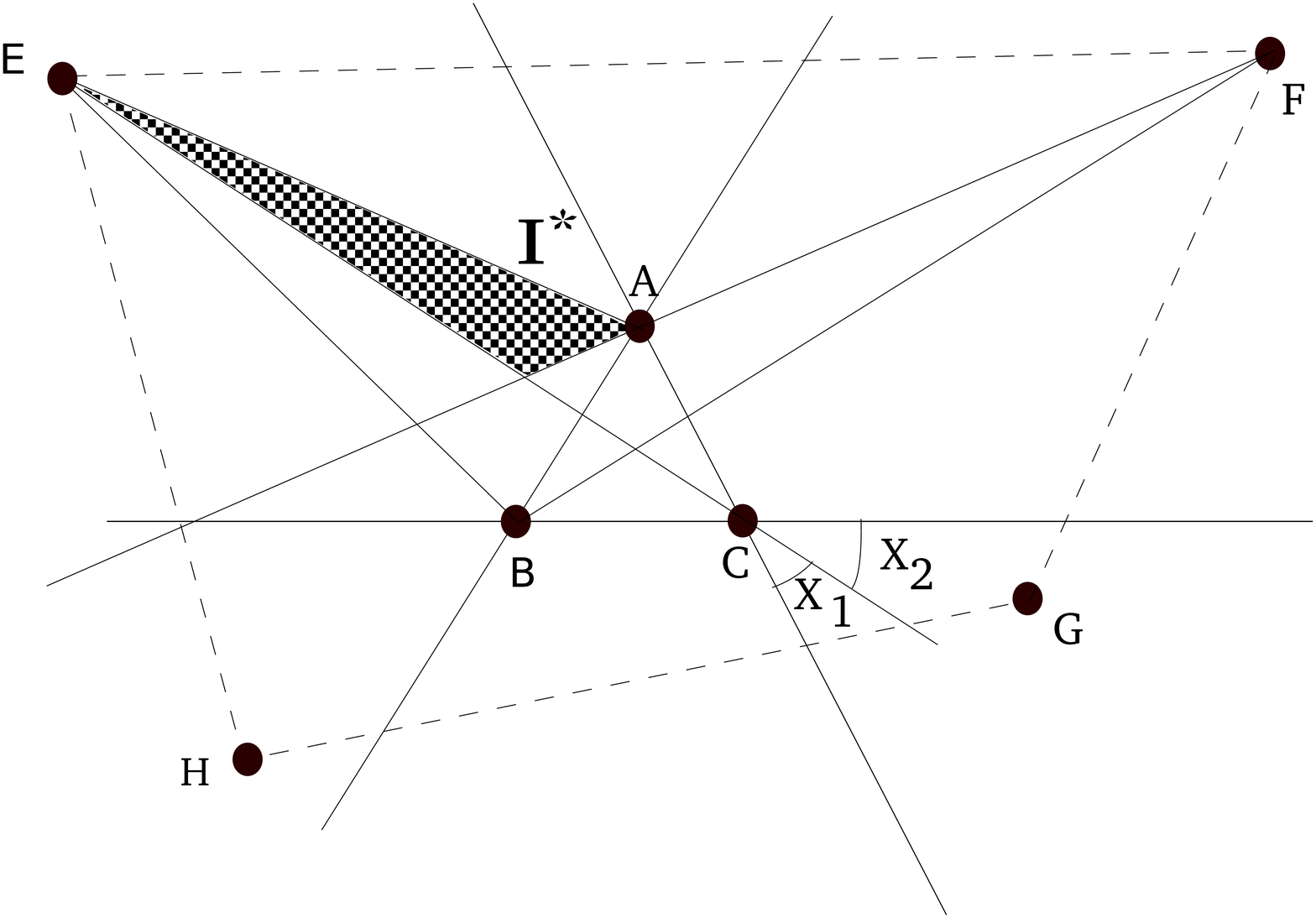}} 
   \hspace{0.1\textwidth}
   \caption{Configuration 7.2 (I,I,O,O) G is in X2 and EAG is an clockwise turn}
\label{lq5}
\end{figure}
\begin{figure}
  \centering
   {\includegraphics[width=0.53\textwidth]{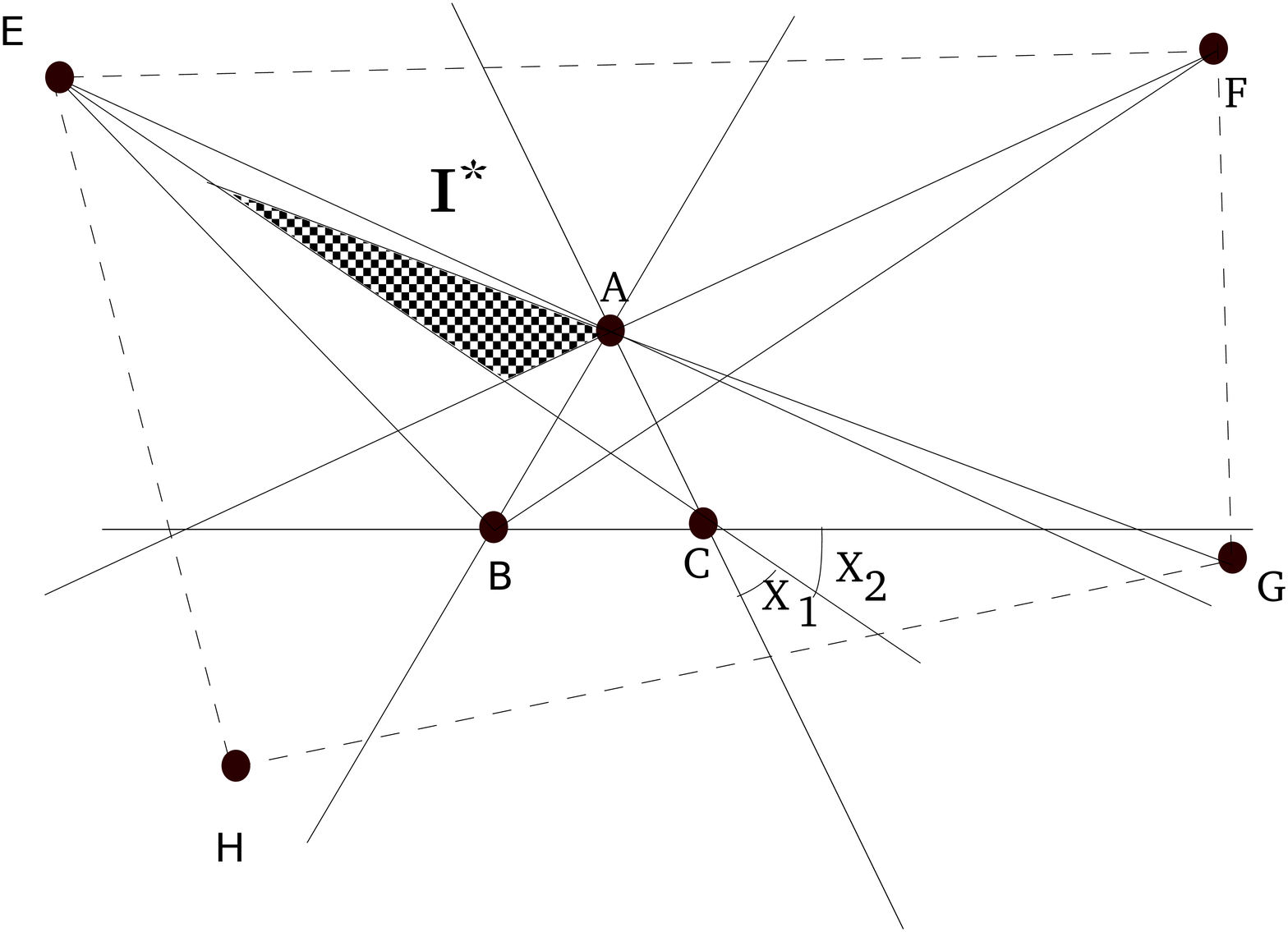}} 
   \hspace{0.1\textwidth}
   \caption{Configuration7.2 (I,I,O,O) G is in X2 and EAG is anticlockwise turn}
\label{lq7}
\end{figure}
(I,O,O,O) (3 points each in different $O$ regions, 1 point in a $I$ region): Let us consider all the empty convex 4-gons of 
this configuration (see $figure$ \ref{label13}).
 The regions that are not covered by the $O$ regions of empty convex 4-gons are precisely the regions where a point is added without forming an empty convex 5-gon. 
We will show that there exists a region $I^{*}$ in these feasible regions where a point is placed forming configuration 8. 
The convex 4-gons that are formed by $U(4,0)$ of configuration 7.2 are not empty, so they are not considered.
 The remaining empty convex 4-gons are of the form  $U(2,2),U(1,3),U(3,1)$ of configuration 7.2. 
It is important to note that $GHE$ should be an anticlockwise turn and $GFE$ a clockwise turn, otherwise we do not have convex 4-gon in the first convex layer.

\begin{figure}
  \centering
   {\includegraphics[width=0.47\textwidth]{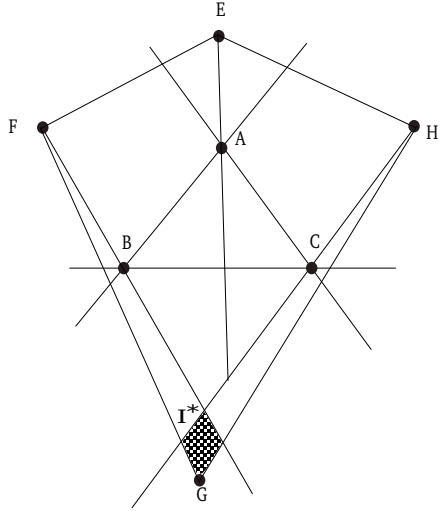}} 
   \hspace{0.2\textwidth}
   \caption{Configuration 7.2 (I,O,O,O)}
    \label{label13}  
\end{figure}
\textit{U(3,1) of configuration 7.2 (I,O,O,O):}
$EHAF$ is the empty convex 4-gon of this type and it has $I^{*}$ in its $I$ region. Thus $I^{*}$ is feasible for these convex 4-gons. 

\textit{U(1,3) of configuration 7.2 (I,O,O,O):}
$FACB,HABC,ACGB$ are the empty convex 4-gons
 of this type. $FACB,HABC$ have $I^{*}$ in there $I$ regions and $ACGB$ 
has $I^{*}$ inside. Thus $I^{*}$ is feasible for these convex 4-gons. 

\textit{U(2,2) of configuration 7.2 (I,O,O,O):}
$EABF,EHCA$ are the empty convex 4-gons of this type and these have $I^{*}$ in there $I$ regions. Thus $I^{*}$ is feasible for these convex 4-gons. 
Any point added to $I^{*}$ results in configuration 8.

\end{proof}
We call a point set as bad configuration for the convex 5-gon game if the point set has no convex 5-gon and any point added to the point set forms an convex 5-gon. 
Similarly, we can define bad configuration for the empty convex 5-gon game. Note that the game ends in the $i$th step only if the point set reached in 
the $i-1$th step is a bad configuration.   

We now argue that the game will always reach the 9th step, i.e., there is no possibility for the game to end earlier.
 To prove this we show that there does not exist point set
 with $2k$ points where $k=2,3$ that are bad configurations. 
 We do not consider point set with $2k+1$ points, $k=2,3$, because it is player 2's turn to form such a point set 
and  even though there are such point sets which are bad configurations, player 2 is able to avoid them in the game by ensuring that configuration 5 is reached
in the 5th step and configuration 7.1 or configuration 7.2 is reached in the 7th step. By lemma \ref{lem1}, \ref{le3} there exists feasible region to 
place points and hence these configurations are not bad configurations.

It is easy to see that no point set with 4 points are bad configurations. 

\begin{lem}
\label{added}
 There are no point set  with 6 points that are bad configurations. 
\end{lem}
\begin{proof}

When the point set contains 6 points it is easy to see that (3,3) and (4,2) are the only
 valid convex layer configurations for the  convex 5-gon game and (3,3), (4,2), (5,1)
are the valid convex layer configurations for the empty convex 5-gon game.
To show that a (3,3), (4,2) convex layer configurations are not bad configurations we show a feasible region where a point is added without 
forming an convex 5-gon.
To show that a (5,1) convex layer configuration is not a bad configuration we show a feasible region where a point is added without 
forming an empty convex 5-gon.
\begin{figure}[ht]
\begin{minipage}[b]{0.53\linewidth}
  \centering
   {\includegraphics[width=0.5\textwidth]{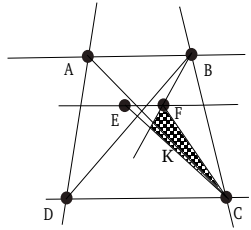}} 
   \hspace{0.1\textwidth}
   \caption{Case 1 of (4,2) configuration}
\label{lastproof21o}
\end{minipage} 
\begin{minipage}[b]{0.53\linewidth}
  \centering
   {\includegraphics[width=0.5\textwidth]{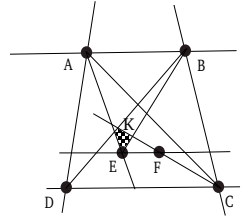}} 
   \hspace{0.1\textwidth}
   \caption{Case 2 of (4,2) configuration}
\label{lastproof21s}
\end{minipage}
\end{figure} 
\begin{figure}
  \centering
   {\includegraphics[width=0.25\textwidth]{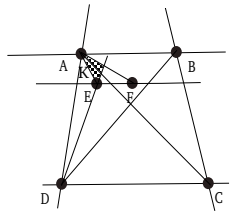}} 
   \hspace{0.1\textwidth}
   \caption{Case 3 of (4,2) configuration}
\label{lastproof21a}
\end{figure} 

\textit{(4,2) Convex layer configuration:} Consider the line joining $EF$  of the second convex layer (see $figure$ \ref{lastproof21o}). This line divides the first convex layer into 2 parts. If either of the parts contains
 3 points of the first convex layer then they form an empty convex 5-gon with the 2 points of the second convex layer.
 So the only other possible option is that both the parts contain 2 points.

Based upon the position of $E,F$ in the diagonal triangles of the first convex layer, (4,2) convex layer configurations are divided into 3 cases.

Case 1:  $E,F$ are in opposite triangles, see $figure$ \ref{lastproof21o}.

Case 2:  $E,F$ are in same triangle, see $figure$ \ref{lastproof21s}.

Case 3:  $E,F$ are in adjacent triangles, see $figure$ \ref{lastproof21a}.

In each of the above cases a feasible region $K$ exists where a point can be added without forming a convex 5-gon (see shaded 
region in $figure$ \ref{lastproof21o}, \ref{lastproof21s}, \ref{lastproof21a}).
Thus (4,2) convex layer configurations are not bad configurations.

\textit{(3,3) Convex layer configuration:}
 We make the following observations on the placement of the three points of the first convex 
layer in the $I$ and $O$ regions of the triangle formed by the second convex layer (see $figure$ \ref{12}). If an $I$ region has 3
 points then the triangle in the second convex layer will not be contained in the triangle of the 
first convex layer. Hence an $I$ region cannot have 3 points. Similarly an $O$ region also cannot have 3 points. If an $O$ region has
 2 points then the opposite $I$ region has to have the third point. In this case, there is an empty convex 5-gon being formed by 
these 2 points with the 3 points of the triangle. Hence an $O$ region cannot have more than 1 point.

\begin{figure}[ht]
\begin{minipage}[b]{0.48\linewidth}
  \centering
   {\includegraphics[width=0.53\textwidth]{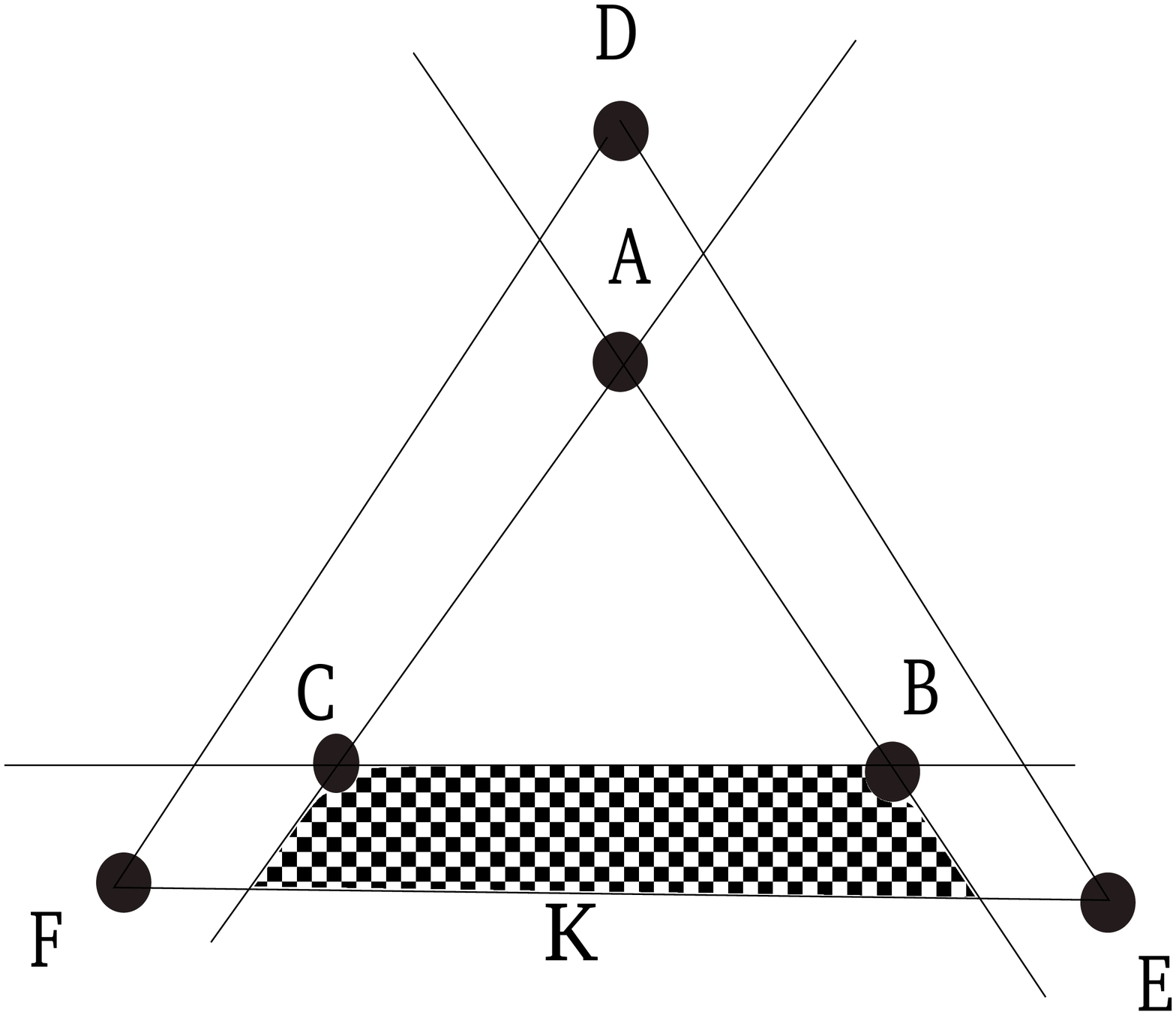}} 
   \hspace{0.1\textwidth}
   \caption{(I,I,I) configuration}
\label{lastproof1}
\end{minipage} 
\begin{minipage}[b]{0.5\linewidth}
  \centering
   {\includegraphics[width=0.6\textwidth]{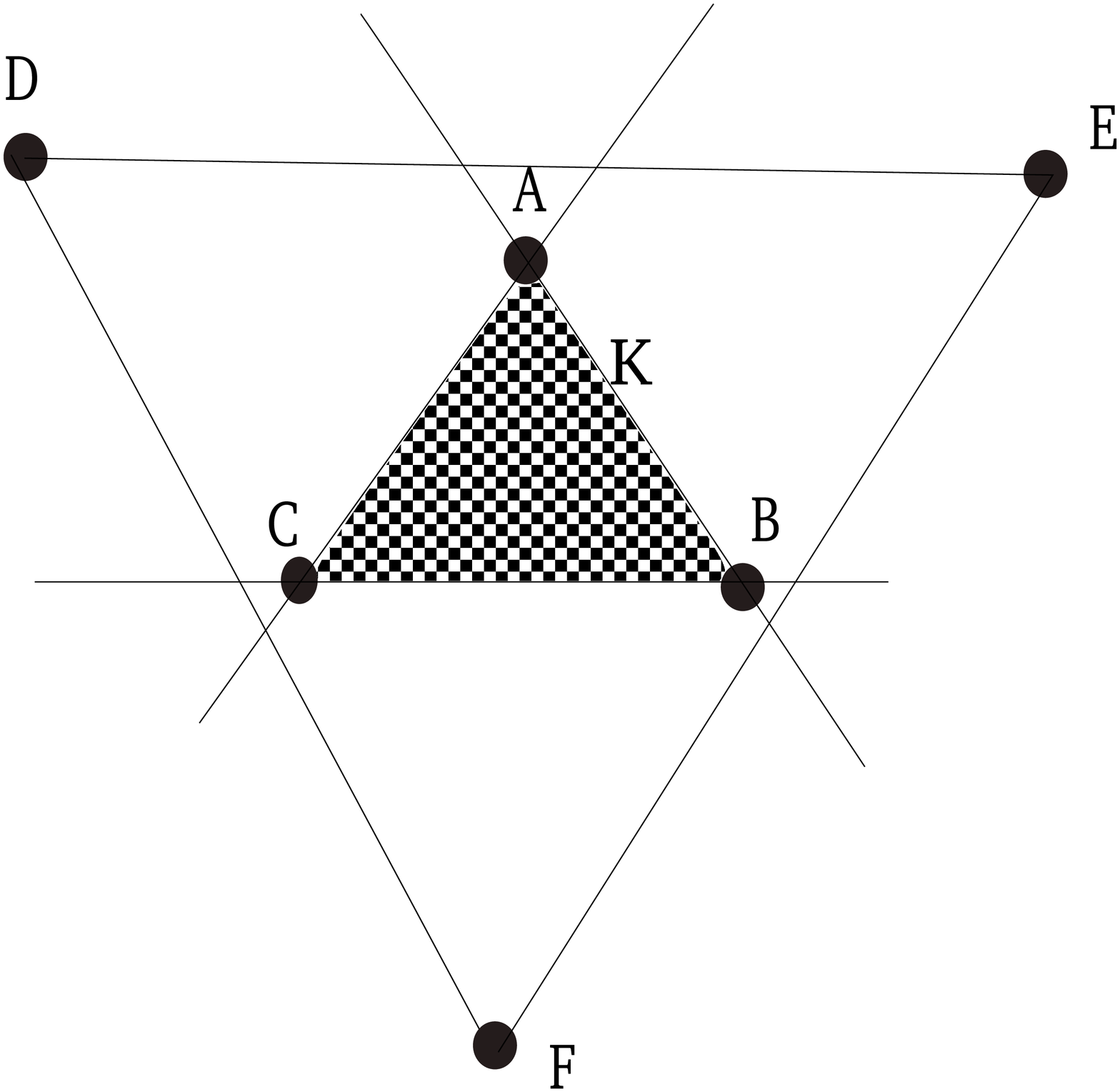}} 
   \hspace{0.1\textwidth}
   \caption{(O,O,O) configuration}
\label{lastproof2}
\end{minipage}
\end{figure} 

\begin{figure}[ht]
\begin{minipage}[b]{0.5\linewidth}
  \centering
   {\includegraphics[width=0.72\textwidth]{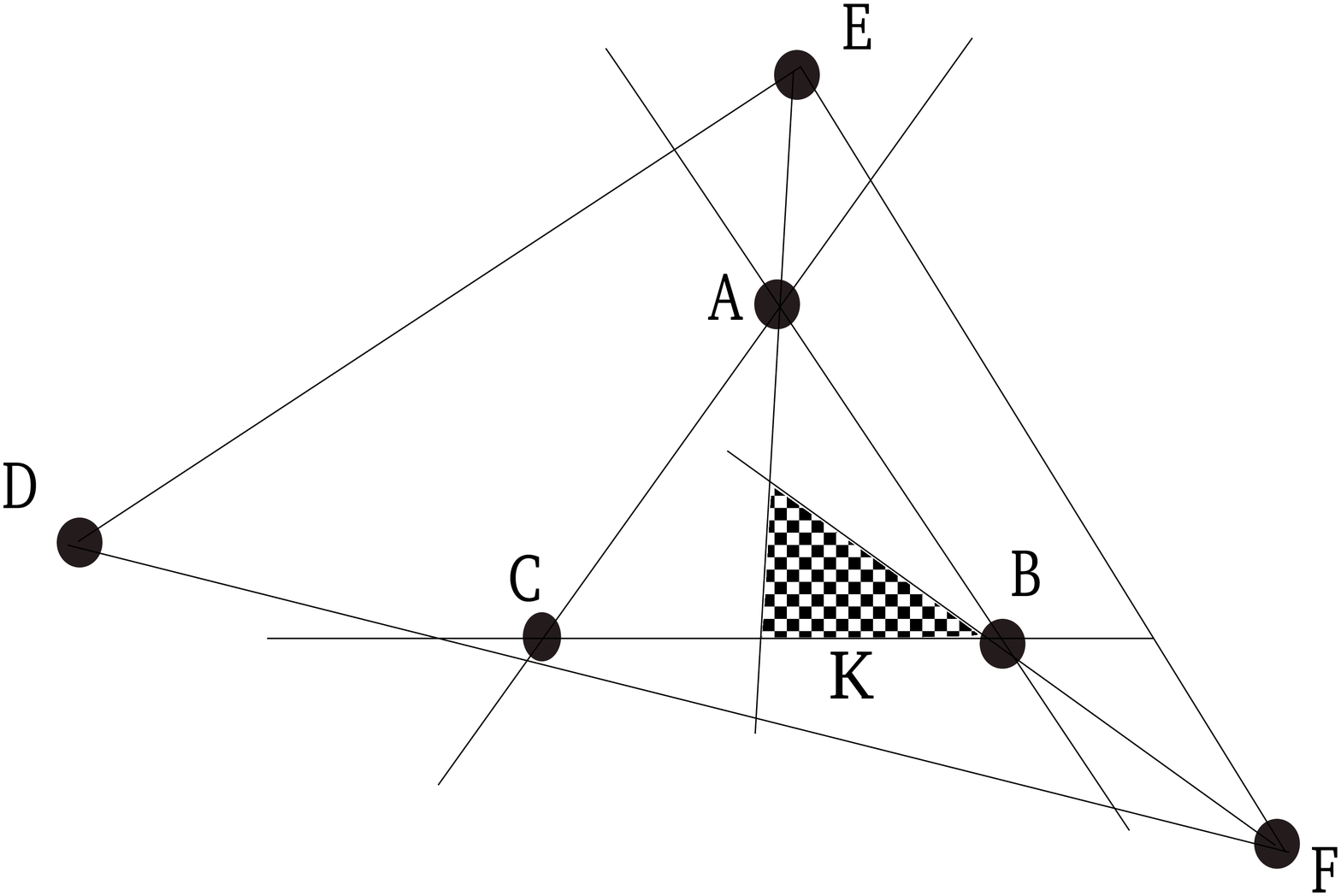}} 
   \hspace{0.1\textwidth}
   \caption{(I,I,O) configuration}
\label{lastproof3}
\end{minipage} 
\begin{minipage}[b]{0.5\linewidth}
  \centering
   {\includegraphics[width=0.72\textwidth]{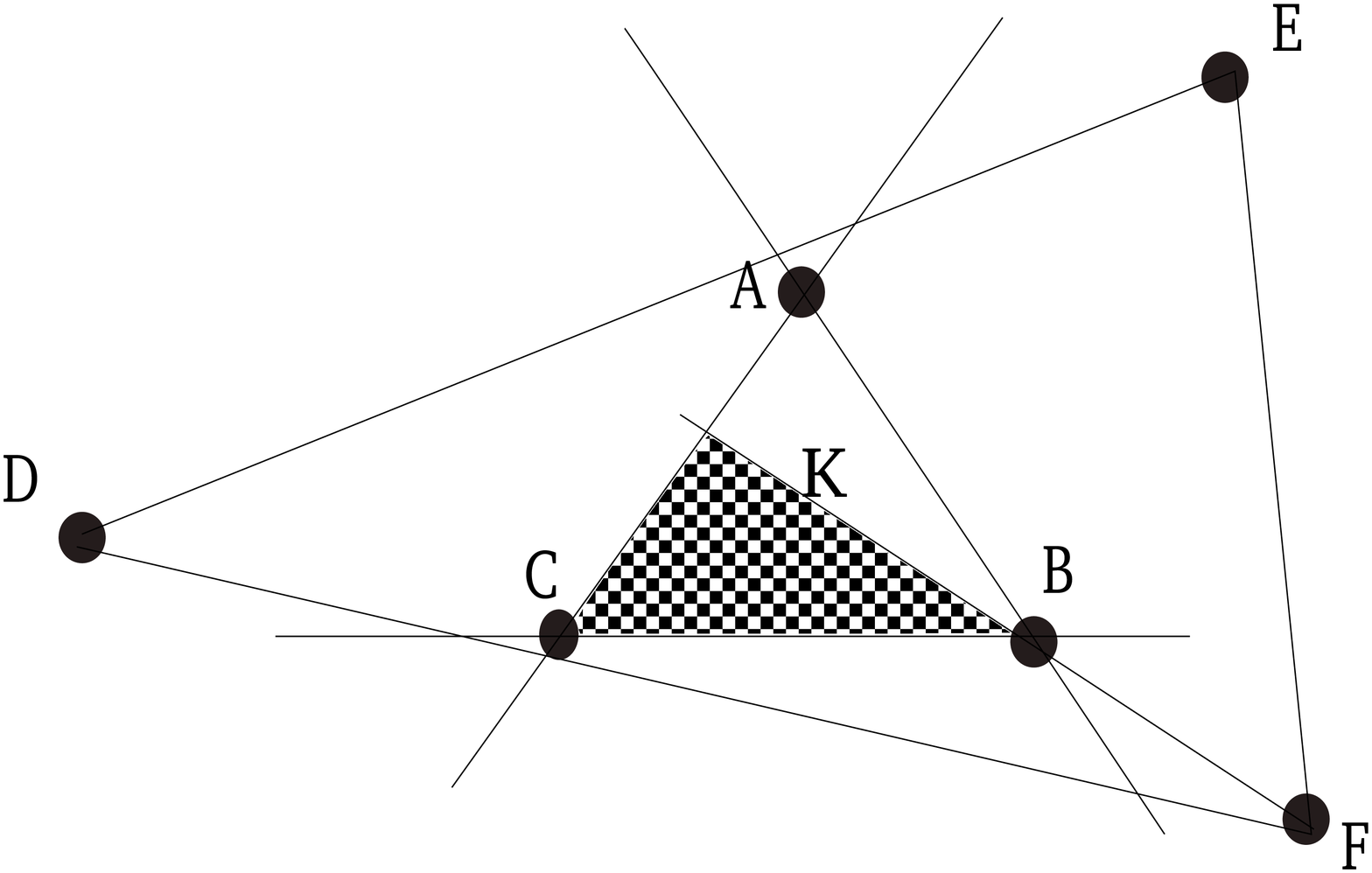}} 
   \hspace{0.1\textwidth}
   \caption{(O,O,I) configuration}
\label{lastproof4}
\end{minipage}
\end{figure}

The only possible (3,3) point configurations are the following:

$(I.I,I)$ (3 points in different $I$ regions) $figure$ \ref{lastproof1}.

$(O,O,O)$ (3 points in different $O$ regions) $figure$ \ref{lastproof2}.

$(I,I,O)$ (2 points are in different $I$ regions and 1 point in an  $O$ region) $figure$ \ref{lastproof3}.

$(O,O,I)$ (2 points are in different $O$ regions and 1 point in an  $I$ region) $figure$ \ref{lastproof4}.

$(2I,O)$ (2 points in an $I$ region and 1 point in the opposite $O$ region) $figure$ \ref{lastproof5}.

In each of the above cases a feasible region $K$ exists where a point can be added without forming a convex 5-gon (see shaded 
region in $figure$ \ref{lastproof1}, \ref{lastproof2}, \ref{lastproof3}, \ref{lastproof4}, \ref{lastproof5}).
Thus (3,3) convex layer configurations are not bad configurations.

\textit{(5,1) Convex layer configuration:}
A feasible region $K$ exists where a point can be added without forming a empty convex 5-gon (see shaded region in $figure$ \ref{lastproof31}). 

\end{proof}

\begin{figure}[ht]
\begin{minipage}[b]{0.49\linewidth}
  \centering
  {\includegraphics[width=0.75\textwidth]{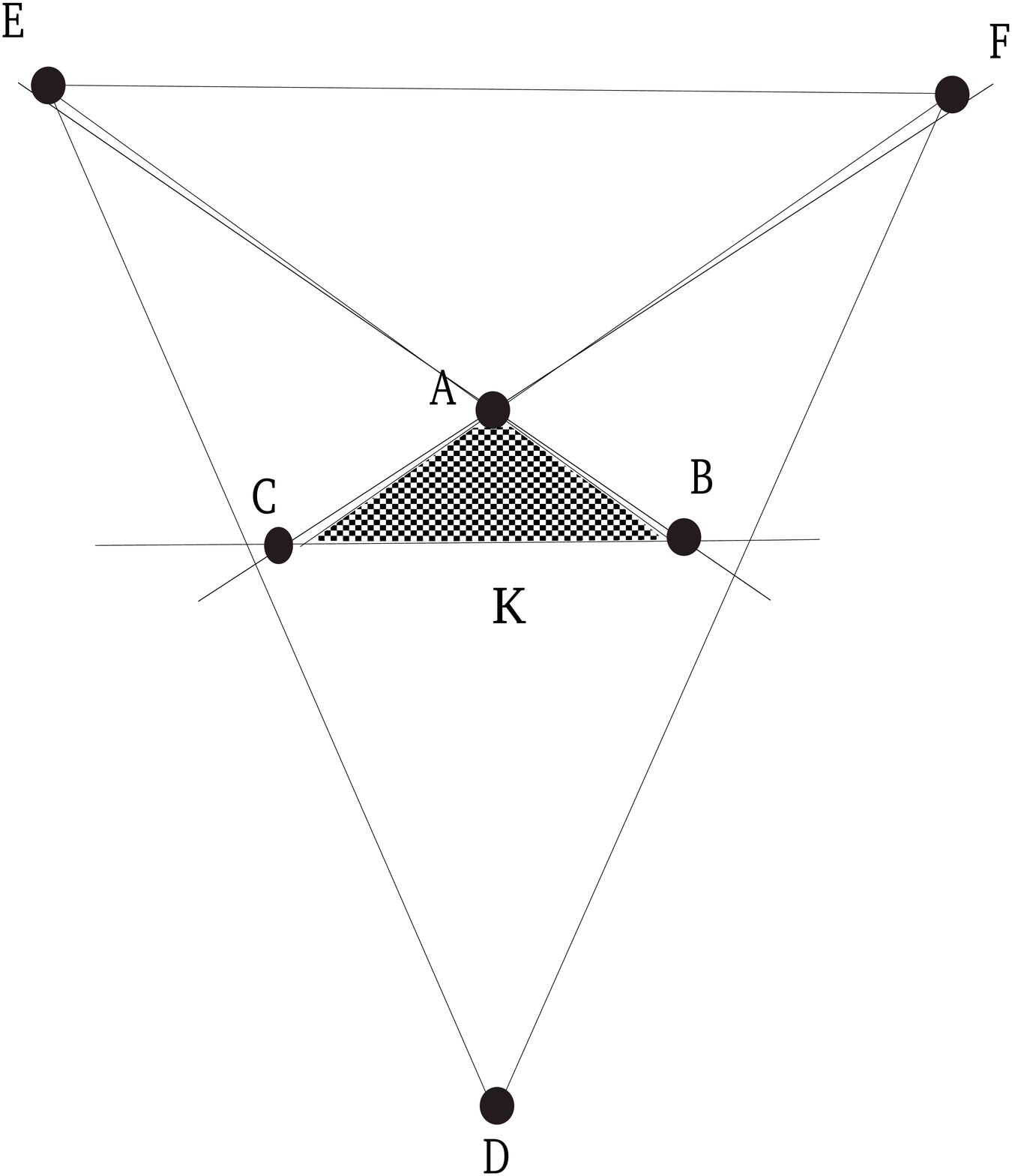}} 
   \hspace{0.1\textwidth}
   \caption{(2I,O) configuration}
\label{lastproof5}
\end{minipage} 
\begin{minipage}[b]{0.51\linewidth}
  \centering
  {\includegraphics[width=1.0\textwidth]{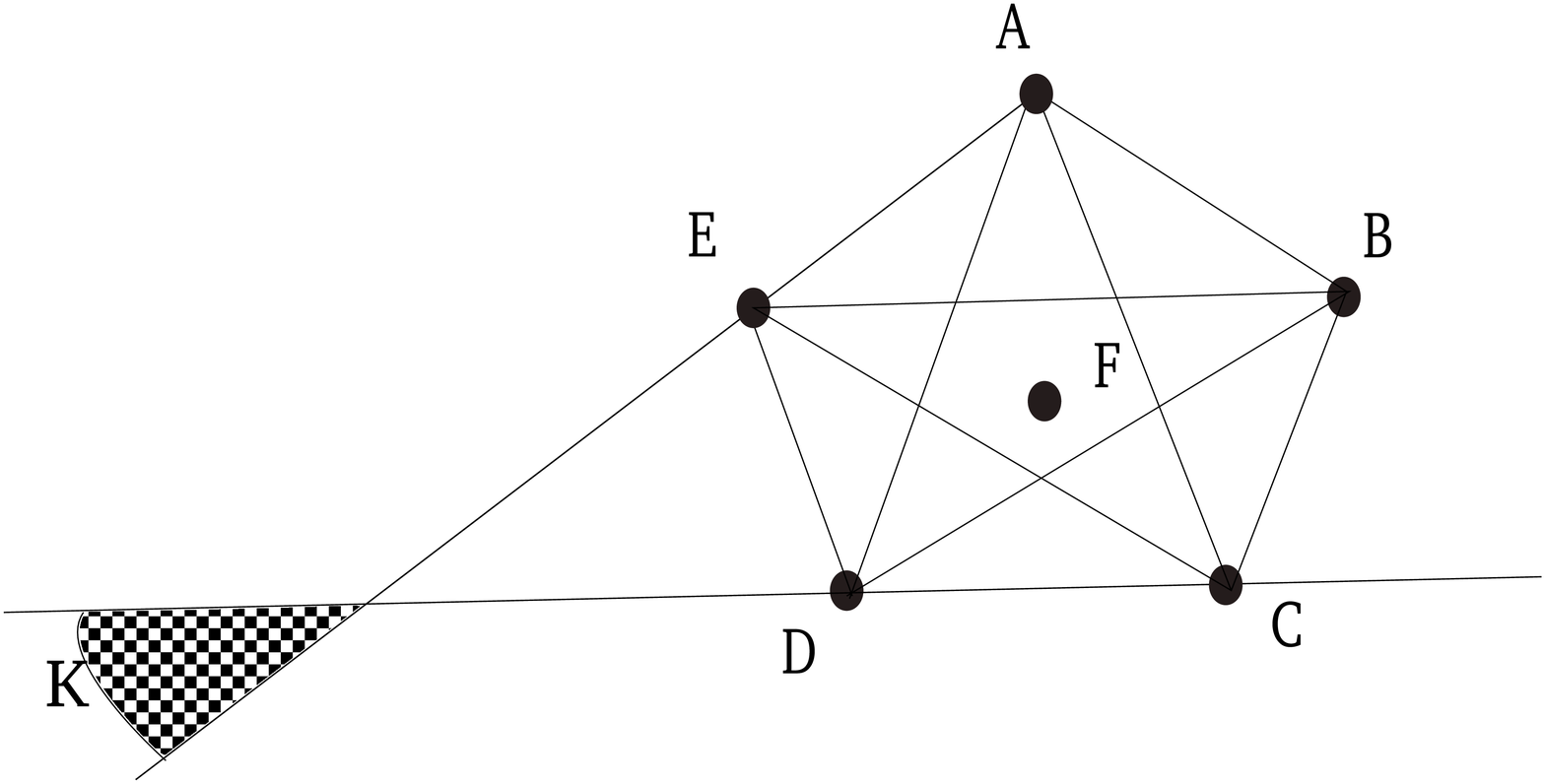}} 
   \hspace{0.1\textwidth}
   \caption{(5,1) configuration}
\label{lastproof31}  
\end{minipage}
\end{figure} 

\begin{thm}
 The convex 5-gon game always ends in the 9th step.
\end{thm}
\begin{proof} The game does not end in the 5th step because no 4 point sets  are bad configurations. By lemma \ref{lem1}, the game will reach 
either configuration 6.1 or 6.2. Since all 6 point sets are not bad configurations(lemma \ref{added}), the game will not end in the 7th step.
 By lemma \ref{lem2} the game will reach either configuration 7.1 or 7.2. By lemma \ref{le3}, the game will reach configuration 8 and finally 
since $N(5)=9$ the game ends in the 9th step and
player 2 wins the game.
\end{proof}

We will now show that any point added to configuration 8 forms an empty convex 5-gon and hence the empty convex 5-gon game also ends
in the 9th step. First we prove the following lemma.

\begin{lem}
\label{lem4} 
A (4,3,2) convex layer configuration has an empty convex 5-gon.
 \end{lem}
\begin{proof} 
Let $F,G,H,J$ be the 4 points of $CH(P)$. Let $A,B,C$ be the points
 in the second convex layer and $D,E$ be the points of the inner most convex layer  (see $figure$ \ref{222}).

The points $F,G,H,J$ lie
 in the union of beams formed by the beams $DE:CB,D:AB$ and $E:AC$. Since beam $DE:CB$
 is type 2, if there exists a point in its beam region then there  exists an empty convex 5-gon.
Thus, one of the type 1 beams $D:AB$ or $E:AC$ beams has atleast 2 points, which forms an empty convex 5-gon.
\end{proof}
\begin{figure}
  \centering
   {\includegraphics[width=0.23\textwidth]{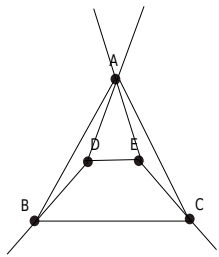}} 
   \hspace{0.1\textwidth}
  \caption{(3,2) configuration}
\label{222}
\end{figure}
\begin{figure}
   \centering
   {\includegraphics[width=0.7\textwidth]{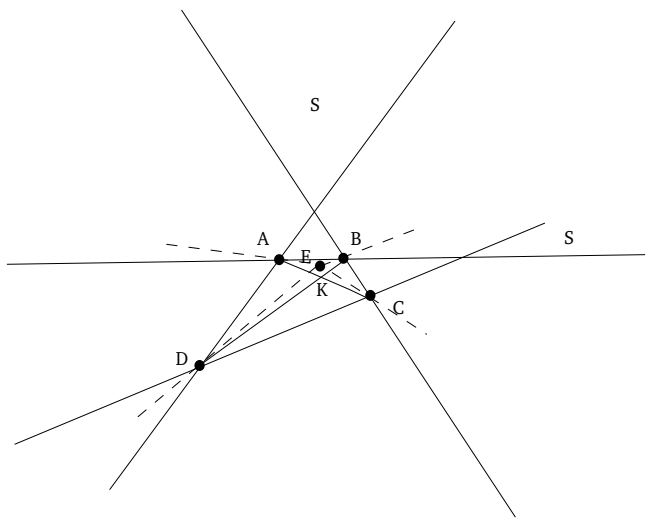}} 
   \hspace{0.5\textwidth}
   \caption{Case 1 for (4,4,1) configuration}
\label{type441}
\end{figure}
\vspace*{0.2cm}
\begin{lem}
\label{lem5} 
A (4,4,1) convex layer configuration has an empty convex 5-gon.
 \end{lem}
\begin{proof} Let $F,G,H,J$ be the points of $CH(P)$. Let $A,B,C,D$ be the points of the second convex layer.
 Let $K$ be the intersection of the diagonals of the convex 4-gon $ABCD$.
 Let $E$ be the point inside the convex 4-gon $ABCD$. The union of beams $E:AB,E:BC,E:CD,E:DA$ contain the points $F,G,H,J$.

Depending upon the position of point $E$ inside the $ABCD$ convex 4-gon the proof is divided into 4 cases.
\begin{figure}
   \centering
   {\includegraphics[width=0.7\textwidth]{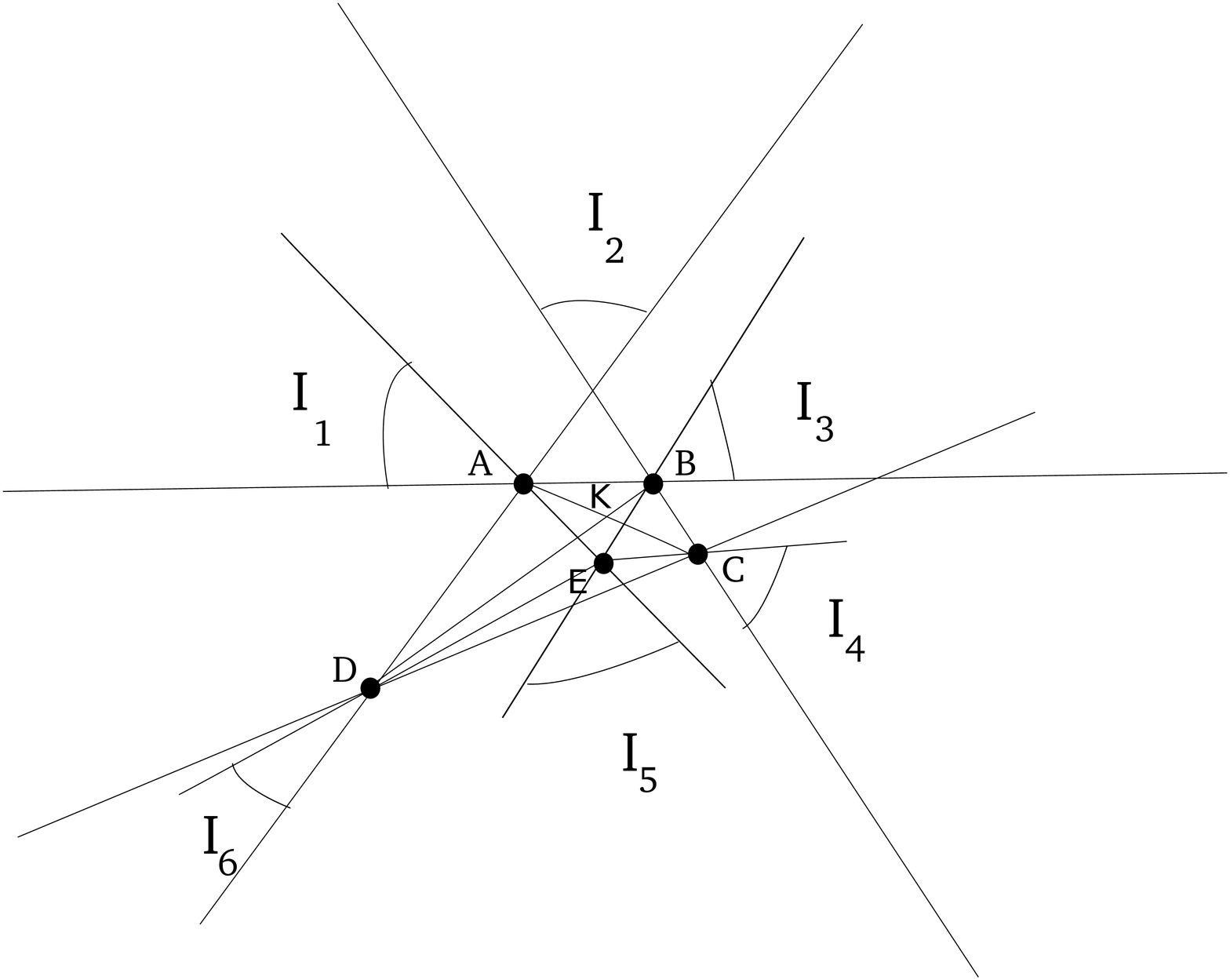}} 
   \hspace{0.5\textwidth}
   \caption{Case 2 for (4,4,1) configuration}
\label{type4412}
\end{figure}

\textbf{Case 1:} When ${E}$ is in the triangle $ABK$ (see $figure$  \ref{type441}).
In this case $E:CD$ beam is covered by the union of $AE:CD$ and $EB:CD$ beams. Therefore any point in the $E:CD$ beam forms an empty convex 5-gon. 
If $E:CD$  beam does not have any point, one of the 3 type 1 beams $E:BC,E:AB,E:DA$ has atleast 2 points (among $F,G,H,J$) which forms an empty convex 5-gon.  

\textbf{Case 2:} When $E$ is in the triangle $KCD$ (see $figure$  \ref{type4412}).
In this case $I_{1},I_{2},I_{3},I_{4},I_{5},I_{6}$ are the only feasible regions outside the $ABCD$ convex 4-gon for the placement of $F,G,H,J$ points. 
The remaining outer regions are present in the $O$ region of some convex 4-gon and hence not feasible.
The $E:CD$ beam contains $I_{4},I_{5},I_{6}$ regions. Since $E:CD$ beam is type 1 beam it can have atmost 1 point. Therefore $I_{1},I_{2},I_{3}$ regions combined should have 3 points.
From the figure we can see that the union of $I_{1},I_{2},I_{3}$ regions has 3 points which forms an empty convex 5-gon with the points $A,B$. 

The case where $E$ is in the triangle $KBC$ is symmetrical to the case when it is in triangle $ABK$.
The case where $E$ is in the triangle $KAD$ is symmetrical to the case when it is in triangle $KCD$.

\end{proof}
Let ${E,F,G,H}$ be the points of the first convex layer and  ${A,B,C,D}$  be the points of the second convex layer of configuration 8. 
Without loss of generality, let us assume that rays $\overrightarrow{AB},\overrightarrow{DC}$ intersect and rays
 $\overrightarrow{DA},\overrightarrow{CB}$ intersect.

 \begin{figure}
   \centering
   {\includegraphics[width=0.6\textwidth]{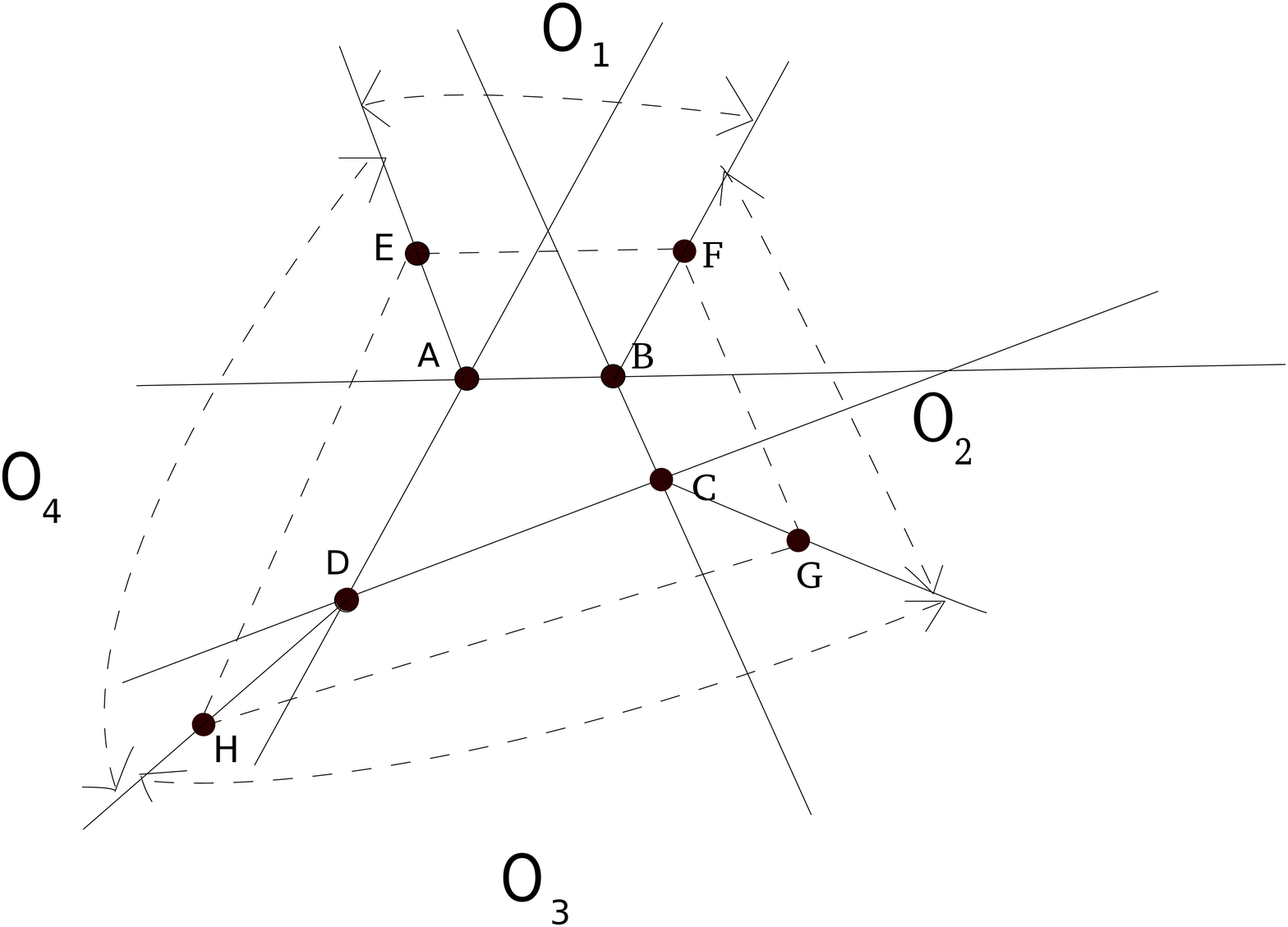}} 
   \hspace{0.5\textwidth}
   \caption{When no pair of rays intersect}
\label{bad44case1}
 \end{figure}

\begin{lem}
\label{lem6} 
The 4 type 2 beams AB : FE,BC : GF,DC : GH,AD : HE  cover the entire outer region in configuration 8.
\end{lem}
\begin{proof} 
We consider 2 cases depending on the intersections of the rays $\overrightarrow{AE},\overrightarrow{BF},\overrightarrow{CG},\overrightarrow{DH}$.

\textbf{Case1:} When none of the rays  $\overrightarrow{AE},\overrightarrow{BF},\overrightarrow{CG},\overrightarrow{DH}$ intersect each other (see $figure$ \ref{bad44case1}).

The  regions  covered by $AB : FE,BC : GF,DC : GH,AD : HE$ beams are correspondingly $O_{1},O_{2},O_{3},O_{4}$ beam regions
which cover the entire outer region in configuration 8.

\textbf{Case2:}  When some of the rays  $\overrightarrow{AE},\overrightarrow{BF},\overrightarrow{CG},\overrightarrow{DH}$ intersect each other (see $figure$ \ref{bad44case2}).
 
Now we prove that atmost one of the pairs of rays $(\overrightarrow{AE},\overrightarrow{BF})$ or $(\overrightarrow{CG},\overrightarrow{BF})$
 can intersect in configuration 8. Assume that $(\overrightarrow{AE},\overrightarrow{BF})$ intersect (see $figure$  \ref{bad44case2}). 
The case $(\overrightarrow{CG},\overrightarrow{BF})$ intersecting is symmetrical.

 Consider $EA$ and $FB$, extend them inside $ABCD$ convex 4-gon until they intersect with $CD$. 
Let the points of intersection be ${J,K}$. When $GC$ is extended inside the $ABCD$ convex 4-gon it has to intersect 
$BK$ before it intersects $AB$ or $AD$. Hence $(\overrightarrow{BF} , \overrightarrow{CG})$ do not 
intersect. By a similar reasoning, ray $\overrightarrow{HD}$ intersects $AJ$ before it intersects $AB$ or $BC$. 
Hence $(\overrightarrow{AE} , \overrightarrow{DH})$ do not intersect. 
 The pair of rays $(\overrightarrow{CG} , \overrightarrow{DH})$ do not intersect because the pair of rays $(\overrightarrow{BC} , \overrightarrow{AD})$
do not intersect.

The  regions  covered by $AB : FE,BC : GF,DC : GH,AD : HE$ beams are correspondingly $O_{1},O_{2},O_{3},O_{4}$ beam regions (see $figure$  \ref{bad44case2})
which cover the entire outer region in configuration 8.
 
\end{proof}
 \begin{figure}
   \centering
   {\includegraphics[width=0.6\textwidth]{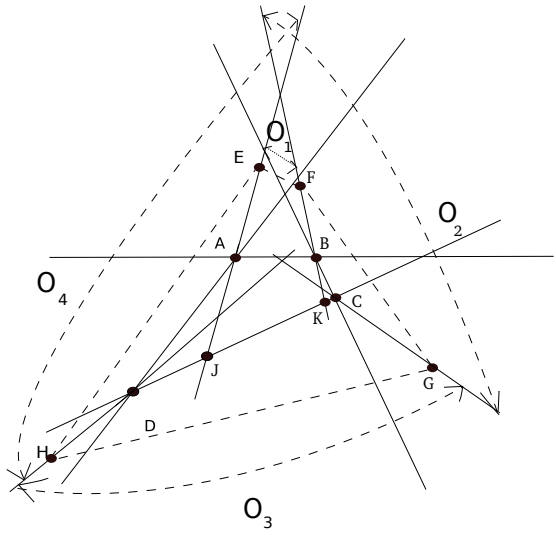}} 
   \hspace{0.5\textwidth}
   \caption{When one pair of rays intersect}
\label{bad44case2} 
\end{figure}
\begin{figure}
  \centering
   {\includegraphics[width=0.6\textwidth]{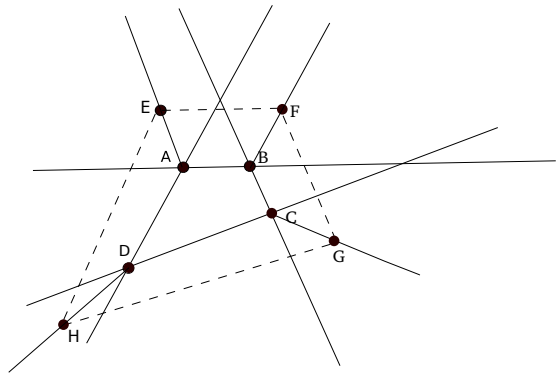}} 
   \hspace{0.1\textwidth}
   \caption{Configuration 8}
\label{y}
\end{figure}
\begin{lem}\label{proof for config8}
Any point added to configuration 8 forms an empty convex 5-gon. 
\end{lem}
\begin{proof}

We consider 2 cases depending on whether the point is added inside $EFGH$ convex 4-gon or outside $EFGH$ convex 4-gon (see $figure$ \ref{y}).
\vspace{1cm}

\textbf{Case1:} Placing the point inside $EFGH$ convex 4-gon:

The point added lies in $I$ region, 
$S$ region, $Z$ region or  $O$ region of $ABCD$ convex 4-gon.
 If the point is placed in the $O$ region of $ABCD$ convex 4-gon then there exists an empty convex 5-gon.
 If point is placed  in the $S$ region of $ABCD$ convex 4-gon then it forms a (4,3,2) convex layer configuration which 
contains an empty convex 5-gon (lemma \ref{lem4}).
 If the point is placed in the $I$ or $Z$ region of the $ABCD$ convex 4-gon then it forms a (4,4,1) convex layer configuration which 
contains an empty convex 5-gon (lemma \ref{lem5}). 

\textbf{Case2:} Placing the point outside the $EFGH$ convex 4-gon:

By lemma \ref{lem6}, the 4 beams $AB:FE,BC:GF,DC:GH,AD:HE$  cover the entire outer region.
The added point lies in one of the 4 type 2 beams and forms an empty convex 5-gon.
\end{proof}

\begin{thm}
 The empty convex 5-gon game always ends in the 9th step.
\end{thm}
\begin{proof} The game does not end in the 5th step because no 4 point sets  are bad configurations. By lemma \ref{lem1}, the game will reach 
either configuration 6.1 or 6.2. Since all 6 point sets are not bad configurations (lemma \ref{added}), the game will not end in the 7th step.
 By lemma \ref{lem2} the game will reach either configuration 7.1 or 7.2. By lemma \ref{le3}, the game will reach configuration 8 and finally 
from lemma \ref{proof for config8} 
the game ends in the 9th step and
player 2 wins the game.
\end{proof}

\section*{Conclusion}
In our paper we have introduced the two player game variant of Erd\H{o}s-Szekeres problem 
and proved that the game ends in the 9th step for the convex 5-gon and empty convex 5-gon game and player 2 wins in both the cases.

One natural question would be to analyze the game for higher values of $k$ i.e. determine $N_{G}(k)$ and $H_{G}(k)$ for $k > 5$.
Our approach will be very tedious for higher values of $k$ as with the increase in the number
of points in the point set, the number of point configurations increases exponentially.

We have shown that configuration 8 is a bad configuration for the empty convex 5-gon game. 
Another natural question is to determine whether there exists bad configurations for $k > 5$.
 More specifically, does there exist point configurations with the property that any point added to this configuration
 forms an empty convex $k$-gon or convex $k$-gon for $k > 5$. A negative result for the above question gives a lower bound for $N_{G}(k)$ and $H_{G}(k)$.

\bibliographystyle{plain}

\end{document}